\documentclass[11pt]{amsart}
\usepackage{amsmath}
\usepackage{graphicx}
\usepackage{bm}
\usepackage{subfig,tikz}
\usepackage{wrapfig}
\usepackage{xspace}
\usepackage{enumerate}
\theoremstyle{plain}
\newtheorem{theorem}{Theorem}[section]
\newtheorem{lemma}[theorem]{Lemma}
\newtheorem{prop}[theorem]{Proposition}

\theoremstyle{definition}

\numberwithin{equation}{section}
\usepackage{bm}
\usepackage{hyperref} %added
\hypersetup{colorlinks=true, linkcolor=blue, citecolor = blue, urlcolor=blue, filecolor=magenta}

\newcommand{\R}{\mathbb R}

\newcommand{\abs}[1]{\left|#1\right|}
\begin{document}

%\title[] {}

%\begin{abstract}

%\end{abstract}
%\maketitle

\title[A Penrose-type inequality with angular momenta]{A Penrose-type inequality with angular momenta for black holes with 3-sphere horizon topology }

\author[Alaee]{Aghil Alaee}
\address{\parbox{\linewidth}{Aghil Alaee\\
		Department of Mathematics, Clark University, Worcester, MA 01610, USA\\
		Center of Mathematical Sciences and Applications, Harvard University,
		Cambridge, MA 02138, USA}}
\email{aalaeekhangha@clarku.edu, aghil.alaee@cmsa.fas.harvard.edu}

\author[Kunduri]{Hari Kunduri}
\address{\parbox{\linewidth}{Hari Kunduri\\
		Department of Mathematics and Statistics, Memorial University of Newfoundland, St John's NL A1C 4P5, Canada}}
\email{hkkunduri@mun.ca}

%\date{\today}

\begin{abstract}We establish a Penrose-type inequality with angular momentum for four dimensional, biaxially symmetric, maximal, asymptotically flat initial data sets $(M,g,k)$ for the Einstein equations with fixed angular momenta and horizon inner boundary associated to a 3-sphere outermost minimal surface.  Moreover, equality holds if and only if the initial data set is isometric to a canonical time slice of a stationary Myers-Perry black hole.
\end{abstract}
\maketitle

\section{introduction}
Inequalities placing a lower bound on the total mass $m$ of a $(n+1)$-dimensional spacetime containing a black hole in terms of the area $A$ of a spatial cross section of the event horizon with $(n-1)$-spherical topology, of the form
\begin{equation}\label{Penrose}
m \geq \frac{1}{2}\left(\frac{A}{\omega_{n-1}}\right)^{\frac{n-2}{n-1}} 
\end{equation} where $\omega_{n-1}$ is the area of canonical $(n-1)$-sphere, are collectively referred to as Penrose inequalities.  Penrose produced a heuristic argument for \eqref{Penrose} based upon the standard picture of gravitational collapse.  The weak cosmic censorship conjecture and final state conjecture play an essential role in this argument, and a counterxample of \eqref{Penrose} would provide strong evidence that at least one of these conjectures is false or incomplete in some way.  A rigorous formulation of the problem, logically independent from these conjectures, is itself a difficult task.  In particular, it is desirable to recast the inequality in terms of quantities which can be determined purely in terms of initial data so that knowledge of the global evolution is not required.  Furthermore, the area of the event horizon, which requires global properties of the spacetime, should be replaced with a suitable proxy that can determined in terms of the quasi-local geometry of the data.  Significant progress in this direction was carried out in the celebrated works of Huisken-Ilmanen~\cite{HI} and shortly thereafter by Bray~\cite{Bray} by using a different approach suited for a more general setting. These authors proved what is known as the Riemannian Penrose inequality, namely \eqref{Penrose} for 3-dimensional asymptotically flat (Euclidean) Riemannian manifolds with non-negative  scalar curvature  with boundary consisting of an outermost minimal surface. Here $m$ is taken to the ADM mass of the initial data, and $A$ the area of this minimal surface, and rigidity is achieved if and only if the initial data is that of a canonical slice of the exterior Schwarzschild black hole spacetime.

Refinements to the Penrose inequality have been established rigorously when the spacetime carries additional conserved charges, such as electric charge~\cite{Khuri:2014wqa}.  Very recently a Penrose type inequality involving angular momentum $J$, was established by Khuri-Sokolowsky-Weinstein for 3-dimensional, asymptotically flat, axisymmetric, maximal initial data sets with outermost minimal boundary~\cite{Khuri:2019pag}. The axial symmetry prevents gravitational waves from carrying angular momenta away to infinity, and hence the angular momentum of the initial and final configurations should be the same.  Thus a heuristic argument suggests a lower bound for the ADM mass in terms of the area $A$ of an outermost, minimal surface and the angular momenta of the initial data, with \eqref{Penrose} recovered in the limit $J=0$.  

The above-mentioned authors rigorously proved a result which is closely related to the desired result~\cite[Theorem 1.1]{Khuri:2019pag}. Namely, they produced a lower bound for the mass in terms of the angular momentum, area of the outermost minimal surface, and an extra term involving a certain integral over the horizon of a quantity associated to the axisymmetric initial data.  The rigidity case is the unique Kerr(-Newmann) black hole with the given $A$ and $J$. The key observation used in the proof is that for 3-dimensional, asymptotically flat, axisymmetric, maximal initial data sets, the ADM mass can be bounded below by a functional which itself can be regarded as a renormalization of the energy of harmonic maps from an auxiliary Euclidean $\mathbb{R}^3\setminus \Gamma$ where $\Gamma$ can be taken to be the $z-$axis.  The geometry of the horizon is encoded in the asymptotic behaviour of the map near the axis.  Such an approach has proved integral in the proof of mass-angular momentum inequalities \cite{Dain} (see the comprehensive review \cite{Dain:2017jkj}), in which case the initial data is taken to be complete initial data sets with two ends; one asymptotically flat and one either asymptotically flat or asymptotically cylindrical.  The new feature is to allow for a minimal surface boundary in the initial data. 

In the present work we will address the analogous problem for 5-dimensional spacetime, in which the natural setting is biaxisymmetric initial data (i.e. the data admits a $U(1) \times U(1)$ isometry subgroup). Previously, in a series of works, we extended the mass-angular momentum inequality with charge to this setting \cite{Alaeemassam,Alaeechargemass,Alaeering,Alaeemassfunc,Alaeelocal,Alaeesmall,Alaeeremark,AlaeeKhuriYau,AlaeeYau}. To motivate the desired inequality, consider the 5-dimensional non-extreme Myers-Perry black hole. This is a stationary, asymptotically flat five-dimensional black hole with spatial horizon cross sections of topology $S^3$.  The solution is characterized by three parameters $(\mathfrak{m}, \mathcal{J}_1, \mathcal{J}_2)$ corresponding to the ADM mass and two independent angular momenta associated to two orthogonal 2-planes of rotation.  Inspection of the explicit solution shows that the conserved quantities satisfy the following relation:
\begin{equation}\label{1}
2\left(\frac{3A}{16\pi}\right)^2
=\frac{8m^3}{3\pi}-\frac{9}{4}(\mathcal{J}_{1}^{2}+\mathcal{J}_{2}^{2})
+\sqrt{\left[\frac{8m^3}{3\pi}-\frac{9}{4}(\mathcal{J}_{1}^{2}+\mathcal{J}_{2}^{2})\right]^2
	-4\left(\frac{9}{4}\right)^{2}\mathcal{J}_{1}^{2}\mathcal{J}_{2}^{2}},
\end{equation}
where $A$ is the black hole's horizon area. Observe that this has the structure of a quadratic formula, namely $\left(\frac{3A}{16\pi}\right)^2$ satisfies the quadratic equation $x^2-bx+c=0$ with
\begin{equation}
b=\frac{8m^3}{3\pi}-\frac{9}{4}(\mathcal{J}_{1}^{2}+\mathcal{J}_{2}^{2}),\text{ }\text{ }\text{ }\text{ }\text{ }
c=\left(\frac{9}{4}\right)^{2}\mathcal{J}_{1}^{2}\mathcal{J}_{2}^{2}.
\end{equation}
Solving for $b$ in the quadratic equation then leads to a formula for the mass in terms of the area and angular momenta of the black hole
\begin{equation}\label{2}
m^3=\frac{3\pi}{8}\left(\frac{3A}{16\pi}\right)^2
+\frac{\frac{3\pi}{8}\left(\frac{9}{4}\right)^{2}\mathcal{J}_{1}^{2}\mathcal{J}_{2}^{2}}
{\left(\frac{3A}{16\pi}\right)^2}
+\frac{27\pi}{32}(\mathcal{J}_{1}^{2}+\mathcal{J}_{2}^{2}).
\end{equation}

Applying Penrose's heuristic argument, based upon the weak cosmic censorship and final state conjectures, we  obtain the conjectured spacetime version of the Penrose inequality with angular momenta, valid for four-dimensional initial data, namely
\begin{equation}\label{3}
m^3\geq\frac{3\pi}{8}\left(\frac{3A}{16\pi}\right)^2
+\frac{\frac{3\pi}{8}\left(\frac{9}{4}\right)^{2}\mathcal{J}_{1}^{2}\mathcal{J}_{2}^{2}}
{\left(\frac{3A}{16\pi}\right)^2}
+\frac{27\pi}{32}(\mathcal{J}_{1}^{2}+\mathcal{J}_{2}^{2})\text{ }\text{ }\text{ }\text{ whenever }\text{ }\text{ }\text{ }A\geq 8\pi\sqrt{\mathcal{J}_{1}\mathcal{J}_{2}}.
\end{equation}
Note that this inequality has been rigorously established by Bray and Lee \cite{Bray:2007opu} when $\mathcal{J}_{1}=\mathcal{J}_{2}=0$, in which case it reduces to the classic Riemannian Penrose inequality (their inequality looks slightly different because they use a different normalization when defining the ADM mass).  Observe that the area-angular momentum inequality is implicitly used in this conjecture, because it can easily be seen that precisely under this condition, the right-hand side of \eqref{2} is monotonically increasing (or nondecreasing) as a function of $A$ with fixed angular momenta. Such a monotonic property is needed in Penrose's heuristic arguments. Moreover, this area-angular momenta inequality has already been established by Hollands for a single black hole \cite{Hollands:2011sy} and extended in \cite{Alaeearea}.  Thus, for a single black hole this is not an extra condition, but for multiple black holes it adds an extra restriction.  Note that if one evaluates the right-hand side at the critical point $A=8\pi\sqrt{\mathcal{J}_{1}\mathcal{J}_{2}}$, corresponding to the extreme Myers-Perry black hole, then \eqref{2} turns into the mass-angular momentum inequality proved in \cite{Alaeemassam} (see also \cite{Alaeechargemass} for a generalization to include charge). Another perhaps more general version of the Penrose inequality can be obtained via the same heuristic arguments, namely
\begin{equation}\label{4}
2\left(\frac{3A}{16\pi}\right)^2
\leq \frac{8m^3}{3\pi}-\frac{9}{4}(\mathcal{J}_{1}^{2}+\mathcal{J}_{2}^{2})
+\sqrt{\left[\frac{8m^3}{3\pi}-\frac{9}{4}(\mathcal{J}_{1}^{2}+\mathcal{J}_{2}^{2})\right]^2
	-4\left(\frac{9}{4}\right)^{2}\mathcal{J}_{1}^{2}\mathcal{J}_{2}^{2}}.
\end{equation}
Note that this one does not need the auxiliary area-angular momentum inequality. The relationship between the two versions of the Penrose inequality is as follows. It should be possible to show that inequality \eqref{4} is algebraically equivalent to the following dichotomy
\begin{align}
\begin{split}
m^3\geq&\frac{3\pi}{8}\left(\frac{3A}{16\pi}\right)^2
+\frac{\frac{3\pi}{8}\left(\frac{9}{4}\right)^{2}\mathcal{J}_{1}^{2}\mathcal{J}_{2}^{2}}
{\left(\frac{3A}{16\pi}\right)^2}
+\frac{27\pi}{32}(\mathcal{J}_{1}^{2}+\mathcal{J}_{2}^{2})\text{ }\text{ }\text{ }\text{ whenever }\text{ }\text{ }\text{ }A\geq 8\pi\sqrt{\mathcal{J}_{1}\mathcal{J}_{2}},\\
m^3\geq& \frac{27\pi}{32}\left(\mathcal{J}_1+\mathcal{J}_2\right)^2\text{ }\text{ }\text{ }\text{ whenever }\text{ }\text{ }\text{ }A\leq 8\pi\sqrt{\mathcal{J}_{1}\mathcal{J}_{2}}.
\end{split}
\end{align}
The first of these inequalities coincides with \eqref{3} and the second one is the mass-angular momentum inequality that we established independent of the area-angular momentum inequality. Furthermore, the lower bound in \eqref{3} is equivalent to the upper bound \eqref{4} and the following lower bound for area, by viewing it as a quadratic equation in $\left(\frac{3A}{16\pi}\right)^2$:
\begin{equation}\label{5}
2\left(\frac{3A}{16\pi}\right)^2
\geq \frac{8m^3}{3\pi}-\frac{9}{4}(\mathcal{J}_{1}^{2}+\mathcal{J}_{2}^{2})
-\sqrt{\left[\frac{8m^3}{3\pi}-\frac{9}{4}(\mathcal{J}_{1}^{2}+\mathcal{J}_{2}^{2})\right]^2
	-4\left(\frac{9}{4}\right)^{2}\mathcal{J}_{1}^{2}\mathcal{J}_{2}^{2}}.
\end{equation}
Thus, the Penrose inequality \eqref{3} is equivalent to \eqref{4} and \eqref{5} with the auxiliary area angular momentum inequality.

We now state our main result. Consider a simply connected asymptotically flat initial data set $(M,g,k)$ consisting of a Riemannian four-manifold $(M,g)$ and a symmetric 2-tensor $k$ satisfying the constraint equations 
\begin{equation}
16\pi \mu = R_g + (\text{Tr}_g k)^2 - |k|^2_g, \qquad 8 \pi J = \text{div}_g ( k - (\text{Tr}_g k)g)
\end{equation} where $(\mu, J)$ represent the energy and momenta densities respectively.  We will assume hereafter that the data is maximal, namely $\text{Tr}_g k =0$ and that the dominant energy condition $\mu \geq |J|$ holds.  Furthermore, we restrict to biaxisymmetric data, which means that the isometry group of $(M,g)$ admits a subgruop isomorphic to $U(1) \times U(1) \equiv U(1)^2$ with no discrete isotropy subgroups,  and that all other quantities associated to the initial data (in particular the second fundamental form $k$ of the embedding) are invariant under this action.   Asymptotic flatness requires that there exists an end $M_{\text{end}} \subset M$ diffeomorphic to $\mathbb{R}^4 \setminus \text{Ball}$ and that there exists $\epsilon >0$ such that in the coordinate chart defined by this diffeomorphism, the data satisfies the decay
\begin{equation}
g_{ij} = \delta_{ij} + O_1 (r^{-1 - \epsilon}), \quad k_{ij} = O(r^{-2-\epsilon}), \quad \mu \in L^2(M_{\text{end}}), \quad J_i \in L^1(M_{\text{end}}).
\end{equation}  Let $\eta_{(i)}$ denote the Killing vector fields generating the $U(1)^2$ action. The asymptotic conditions and biaxismmetry imply that the ADM angular momenta
\begin{equation}
\mathcal{J}_{(k)} = \frac{1}{8\pi} \int_{S_\infty} (k_{ij} - (\text{Tr}_g k) g_{ij}) \nu^i \eta^j_{(k)} \; dS
\end{equation} is well defined, provided $J(\eta_{(i)}) \in L^1(M_{\text{end}})$. 

Our proof relies on the existence of a `generalized Weyl coordinate system', which is a global system of cylindrical-type coordinates in which the class of biaxisymmetric metrics take the form 
\begin{equation}\label{GBmetric}
	g=\frac{e^{2U+2\alpha}}{2\sqrt{\rho^2+(z-\mathfrak{m})^2}}\left(d\rho^2+d z^2\right)+e^{2U}\lambda_{ij}\left(d\phi^i+A^i_l d y^l\right)\left(d\phi^j+A^j_l d y^l\right),
\end{equation}
for smooth functions $U$, $\alpha$, $A_{l}^{i}$, $\omega^i$ and a symmetric positive definite matrix $\lambda=(\lambda_{ij})$ with $\det\lambda=\rho^{2}$, $i,j,l=1,2$, $(y^1,y^2)=(\rho,z)$. All these quantities are independent of the angular coordinates $(\phi^{1},\phi^{2})$ and satisfying the asymptotics \eqref{F1}-\eqref{F12}. Without lose of generality we assume $\mathfrak{m}\geq 0$. Moreover, the coordinates should take values in the following ranges $\rho\in[0,\infty)$, $z \in \mathbb{R}$, and we have the identifications $\phi^i\sim \phi^i + 2\pi$, $i=1,2$.  Such a system of coordinates can be shown to always exist for stationary, biaxisymmetric asymptotically flat solutions of the {\it spacetime} vacuum Einstein equations \cite{Hollands:2007aj}, in which case \eqref{GBmetric} would describe the induced metric on a canonical time slice.  Although this may seem to be a restrictive assumption on the initial data,  note that an analogous system can be shown to exist in the $3+1$-dimensional case.  It is reasonable to expect our class of biaxisymmetric, simply connected data can be cast in the form \eqref{GBmetric}. 

To incorporate the presence of the black hole horizon in our maximal initial data set, the quantity $A$ is taken to be the area of the outermost minimal surface (because apparent horizons reduce to minimal surface in the case of maximal data). Hence we assume $(M,g)$ is a manifold with a single component minimal surface boundary. The boundary represents the black hole horizon which, under the above symmetry hypothesis, must have topology $S^3$ ( or its quotient $L(p,q)$), $S^1\times S^2$, and connected sum of these two cases~\cite{Galloway:2005mf, Hollands:2007aj}.

In the Weyl coordinates, the $U(1)^2$ action degenerates on the axis set $\rho =0$. As explained below, the two symmetry axes of the asymptotic $S^3$ correspond to semi-infinite intervals along the $z-$axis.  In contrast to the `generalized Brill' system' in Weyl coordinates,  the horizon corresponds to a finite interval on the $z-$ axis upon which the functions $(U,\alpha, V)$ exhibit certain singular behaviour which can be modelled by the behaviour of the explicitly known functions associated to Schwarzschild initial data $(U_S, \alpha_S, V_S)$. The `regularized' difference $\bar{U} = U - U_S,\bar\alpha = \alpha - \alpha_S, \bar{V} = V - V_S$ are uniformly bounded with bounded first derivatives.  In our argument, which involves an integration by parts of the mass formula, the combination $\beta:=2\bar\alpha + 6 \bar U - \text{sgn}(z)\bar{V}$ appears.  We may now state our main result as follows:
\begin{theorem}\label{mainthm}
	Let $(M,g,k)$ be a smooth, asymptotically flat, bi-axially symmetric, maximal initial data set for the five-dimensional Einstein equations satisfying $\mu\geq 0$
	and $J(\eta_{(l)})=0$, $l=1,2$ and with non-empty minimal surface boundary such that $A\geq 8\pi\sqrt{\mathcal{J}_{1}\mathcal{J}_{2}}$. Suppose $M^{4}$ is a manifold diffeomorphic to $\mathbb{R}^{4}\setminus \text{Ball}$ with a minimal surface boundary with $S^3$ topology.  Assume $(M,g)$ admits a global system of Weyl coordinates. Let $A_{MP}$ and $\beta_{MP}$ denotes horizon area and Weyl coordinate function for the unique Myers-Perry black hole sharing the same angular momenta and horizon rod length as the initial data set. Then 
	\begin{equation}\label{mainineq}
	m\geq \left(\frac{3\pi}{8}\left(\frac{3A_{MP}}{16\pi}\right)^2
	+\frac{\frac{3\pi}{8}\left(\frac{9}{4}\right)^{2}\mathcal{J}_{1}^{2}\mathcal{J}_{2}^{2}}
	{\left(\frac{3A_{MP}}{16\pi}\right)^2}
	+\frac{27\pi}{32}(\mathcal{J}_{1}^{2}+\mathcal{J}_{2}^{2})\right)^{1/3}+\frac{\pi}{4}\int_{-\mathfrak{m}}^{\mathfrak{m}}\left(\beta(0,z)-\beta_{MP}(0,z)\right)
	\end{equation}
	Moreover equality holds if and only if $(M,g,k)$ is isometric to the canonical slice of the corresponding Myers-Perry spacetime. 
\end{theorem}

Note that if the initial data set $(M,g,k)$ has the same Weyl coordinate functions as the associated Myers-Perry black hole, i.e., $\beta(0,z)=\beta_{MP}(0,z)$, then this result reduces to a proof of the Penrose inequality conjecture \eqref{3}. Our result represents a generalization of the Penrose-type inequality with angular momentum established recently for three-dimensional axisymmetric initial data sets  \cite{Khuri:2019pag}.  It should be noted that the Riemannian Penrose inequality, which holds up to dimension 7, tacitly assumes a minimal surface boundary of spherical topology. A natural question is whether a similar inequality would hold for initial data with a black hole boundary of topology $S^1 \times S^2$.  Indeed, a mass-angular momenta inequality for black ring initial data was proved  in~\cite{Alaeering}, and a Penrose-type inequality might be expected along heuristic lines when considering gravitational collapse to an stationary black ring spacetime.  However, the three-parameter family of black ring spacetime solutions is expected to be dynamically unstable.  Indeed an analysis of perturbations of black rings suggests that a Penrose-type inequality with angular momenta adapted to this setting is unlikely to hold~\cite{Figueras:2011he}. 

\section{The mass functional and reduced energy}\label{Sec1}
 We first record the appropriate asymptotics in three different regions, namely in the asymptotically flat region, the horizon, and near the axis of symmetry. The particular decay rates are motivated in general by the indicated asymptotically flat geometries, and by the desire for certain coefficients, including $\lambda_{ij}$ and $A^{i}_{l}$, to not yield a direct contribution to the ADM mass. In what follows $a,b$ are functions of only $\theta$, $\kappa>0$ is a constant, and ${\sigma}={\sigma}_{ij}d\phi^{i}d\phi^{j}$ is a Riemannian metric on the torus $T^2$ depending only on $\theta$. We begin with the designated asymptotically flat end characterized by $r\rightarrow\infty$.  Note that the Euclidean metric on $\mathbb{R}^4$ in this coordinate system is given by
 \begin{equation}
 \delta_4 = \frac{d\rho^2 + dz^2}{2\sqrt{\rho^2 + z^2}} + \sigma_{ij} d\phi^i d\phi^j = dr^2 + r^2 d\theta^2 + r^2 (\sin^2\theta (d\phi^1)^2 + \cos^2\theta( d\phi^2)^2)
 \end{equation} where standard spherical coordinates $r \in [0,\infty)$ and $\theta \in [0,\pi/2]$ are related to the Weyl coordinates via
 \begin{equation}
 \rho = \frac{r^2}{2}  \sin^2 2\theta, \qquad z = \frac{r^2}{2} \cos 2\theta.
 \end{equation} Near the asymptotically flat end, we require that the functions appearing in the metric \eqref{GBmetric} satisfy
\begin{equation}\label{F1}
	U=O_{1}(r^{-1-\kappa}),\text{ }\text{ }\text{ }\text{ }\text{ }
	\alpha=O_1(r^{-1-\kappa}),\text{ }\text{ }\text{ }\text{ }\text{ }
	A_{\rho}^i=\rho O_{1}(r^{-5-\kappa}),
	\text{ }\text{ }\text{ }\text{ }\text{ }
	A_z^i=O_{1}(r^{-3-\kappa}),
\end{equation}
\begin{equation}\label{F2}
	\lambda_{ii}=\left(1+(-1)^{i}ar^{-1-\kappa}+O_{1}(r^{-2-\kappa})\right)\sigma_{ii},\text{ }\text{ }\text{ }\text{ }\text{ }\text{ }\text{ }\lambda_{12}=\rho^2 O_1(r^{-5-\kappa}),
	\text{ }\text{ }\text{ }\text{ }\text{ }\text{ }\text{ }|k|_{g}=O(r^{-2-\kappa}).
\end{equation}  The other two important regions of interest are the set $\Gamma = \{\rho=0, z\in \mathbb{R}\}$. This is subdivided into two regions. The first are the {\it axes} $\Gamma_\pm = \{ \rho =0, \pm z > \mathfrak{m}\}$ for some $\mathfrak{m} > 0$ and the {\it horizon rod} $H=\{ \rho =0, -\mathfrak{m} < z < \mathfrak{m}\} $. On the horizon some of the functions appearing in \eqref{GBmetric} have a (prescribed) singular behaviour modelled on Schwarzschild initial data (see Appendix A). In particular,  the asymptotics as $\rho\rightarrow 0$ on  $H$ are given by
\begin{equation}\label{F7}
	U=-\frac{1}{2}\log\rho+O_{1}(1),\text{ }\text{ }\text{ }\text{ }\text{ }\text{ }
	\alpha=\frac{1}{2}\log\rho+O_1(1),
\end{equation}
\begin{equation}\label{F8}
	\lambda_{11},\lambda_{22}=O(1),\quad \lambda_{12}=O(\rho),\quad A_{\rho}^i= O_{1}(\rho),
	\text{ }\text{ }\text{ }\text{ }\text{ }\text{ }
	A_z^i=O_{1}(1),\text{ }\text{ }\text{ }\text{ }\text{ }\text{ }
	|k|_{g}=O(1),
\end{equation} On the symmetry axes $\Gamma_\pm$, on which the Killing vector fields $\partial_{\phi^1}$ and $\partial_{\phi^2}$ respectively degenerate, the asymptotics as $\rho\rightarrow 0$ are 
\begin{equation}\label{F10}
	U=O_{1}(1),\text{ }\text{ }\text{ }\text{ }\text{ }\text{ }
	\alpha=O_1(1),\text{ }\text{ }\text{ }\text{ }\text{ }\text{ }
	A_{\rho}^i= O_{1}(\rho),
	\text{ }\text{ }\text{ }\text{ }\text{ }\text{ }
	A_z^i=O_{1}(1),\text{ }\text{ }\text{ }\text{ }\text{ }\text{ }
	|k|_{g}=O(1),
\end{equation}
\begin{equation}\label{F12}
\lambda_{22},\lambda_{12}=O(\rho^{2}),\text{ }\text{ }\text{ }\text{ }\text{ }
	\lambda_{11}=O(1)
	\text{ }\text{ }\text{ on }\text{ }\text{ }\Gamma_-,\text{ }\text{ }\text{ }\text{ }\text{ }
	\lambda_{11},\lambda_{12}=O(\rho^{2}),\text{ }\text{ }\text{ }\text{ }\text{ }
	\lambda_{22}=O(1)
	\text{ }\text{ }\text{ on }\text{ }\text{ }\Gamma_+.
\end{equation}It should be pointed out that regularity of the geometry along the axis implies a compatibility condition between $\alpha$ and $\lambda$. To see this, let $\vartheta\in(-\infty,2\pi)$ be the
cone angle deficiency coming from the metric \eqref{GBmetric} at the axes of rotation, that is
\begin{equation}
	\frac{2\pi}{2\pi-\vartheta}=\lim_{\rho\rightarrow 0}\frac{2\pi\cdot\mathrm{Radius}}{\mathrm{Circumference}}
	=\lim_{\rho\rightarrow 0}\frac{{\displaystyle \int_{0}^{\rho}}\sqrt{\frac{e^{2U+2\alpha}}{2\sqrt{\rho^{2}+(z-\mathfrak{m})^{2}}}
			+e^{2U}\lambda_{ij}A_{\rho}^{i}A_{\rho}^{j}}d\rho}
	{\sqrt{e^{2U}\lambda_{ii}}}=\frac{e^{\alpha(0,z)}}{\sqrt{2|z-\mathfrak{m}|}}\lim_{\rho\rightarrow 0}
	\frac{\rho}{\sqrt{\lambda_{ii}}}
\end{equation}
where $i=1,2$. The cone angle deficiency should vanish $\vartheta=0$, since $(M,g)$ is smooth across the axis, and thus
\begin{equation}\label{cone}
		\alpha(0,z)=\frac{1}{2}\log\left(|z-\mathfrak{m}|\partial_{\rho}^{2}\lambda_{22}(0,z)\right)\quad\text{on $\Gamma_-$},\quad \alpha(0,z)=\frac{1}{2}\log\left(|z-\mathfrak{m}|\partial_{\rho}^{2}\lambda_{11}(0,z)\right)\quad\text{on $\Gamma_+$}\,.
	\end{equation}
As shown in \cite{Alaeemassam}, since the matrix $\lambda$ has fixed determinant, $\det\lambda=\rho^{2}$, there are only two independent functions contained in $\lambda$ which we parameterize as follows
	\begin{equation}\label{inverse}
	V=\frac{1}{2}\log\left(\frac{\lambda_{11}\mu^\mathfrak{m}_+}{\lambda_{22}\mu^\mathfrak{m}_-}\right),
	\qquad W=\sinh^{-1}\left(\frac{\lambda_{12}}{\rho}\right),
	\end{equation}where $\mu^\mathfrak{m}_{\pm}=\sqrt{\rho^2+(z-\mathfrak{m})^2}\pm (z-\mathfrak{m})$. Combining with \eqref{F8}, it follows that $W=O(\rho)$ on $\Gamma_+\cup \Gamma_-$ and $W=O(1)$ on $H$. Together with \eqref{cone}, it shows that
	\begin{eqnarray}\label{V}
	2\alpha (0,z)=-{V}(0,z)\quad \text{on $\Gamma_-$},\qquad 2\alpha (0,z)=V(0,z)\quad \text{on $\Gamma_+$}
	\end{eqnarray} Following the computation given in \cite{Alaeeremark}, the scalar curvature of the metric \eqref{GBmetric} is
\begin{equation}\label{SCALAR}
	\begin{split}
	e^{2U+2\alpha-\log\left(2\sqrt{\rho^2+(z-\mathfrak{m})^2}\right)}R&=-6\Delta U-2\Delta_{\rho,z}\alpha-6|\nabla U|^2
	-\frac{1}{2}|\nabla V|^2-\frac{1}{2}|\nabla W|^2\\&-\frac{1}{2}\sinh^{2}W\abs{\nabla\left(V+h_2\right)}^2
	-\delta_{3}(\nabla h_2,\nabla V)\\
	&-\frac{1}{4}e^{-2\alpha+\log\left(2\sqrt{\rho^2+(z-\mathfrak{m})^2}\right)}\lambda_{ij}(A_{\rho,z}^{i}-A_{z,\rho}^{i})
	(A_{\rho,z}^{j}-A_{z,\rho}^{j}),
	\end{split}
\end{equation}
where $\Delta$ and the norm $|\cdot|$ are with respect to the following flat metric $	\delta_{3}=d\rho^{2}+d z^{2}+\rho^{2} d\phi^{2}$ on an auxiliary $\mathbb{R}^{3}$ in which all quantities are independent of the new variable $\phi\in[0,2\pi]$, $\Delta_{\rho,z}$ is with respect to the flat metric
$\delta_{2}=d\rho^{2}+dz^{2}$ on the orbit space, and $h_{1}=\frac{1}{2}\log\rho$ and $h_{2}=\frac{1}{2}\log\left(\frac{\mu^\mathfrak{m}_-}{
	\mu^\mathfrak{m}_+}\right)$ are harmonic functions on $(\mathbb{R}^{3}\setminus\Gamma,\delta_{3})$.
The ADM mass of the metric \eqref{GBmetric} is \cite{Alaeeremark} 
\begin{equation}\label{massADM}
	m=-\frac{1}{8}\lim_{r\to\infty}\int_{S^3_r}\left(\partial_r\left[6U+2\alpha\right]-4r^{-1}\alpha \right)\, dS_r
\end{equation}where $dS_r=r^3\sin\theta\cos\theta\,d\theta d\phi^1 d\phi^2$ is the volume element of a canonical round 3-sphere $S^3_r$ with radius $r$.
Integrating the first two terms on the right hand of \eqref{SCALAR} over $\R^3$ gives
\begin{equation}\label{mass1}
\begin{split}
\frac{1}{8}\int_{\R^3}-6\Delta U-2\Delta_{\rho,z}\alpha\,\, dx
	%&=-\frac{1}{8}\int_{\R^3}\Delta\left(6 U+2\alpha\right) dx +\frac{1}{8}\int_{\R^3}2\frac{\partial_\rho \alpha}{\rho}dx\nonumber\\
%&=-\frac{1}{8}\int_{\R^3}\Delta\left(6 U+2\alpha\right) dx +\frac{1}{8}\int_{\R^3}2\nabla\log\rho\cdot\nabla\alpha dx\nonumber\\
%&=\frac{1}{8}\lim_{\epsilon\to 0}\int_{\rho=\epsilon}\partial_{\rho}\left(6 U+2\alpha\right)-\frac{1}{8}\lim_{r\to\infty}\int_{S_r}\partial_r\left(6 U+2\alpha\right) \nonumber\\
%&-\frac{1}{8}\lim_{\epsilon\to 0}\int_{\rho=\epsilon}2\alpha\partial_\rho\log\rho +\frac{1}{8}\lim_{r\to 0}\int_{S_r}2\alpha\partial_r\log\rho \nonumber\\
=&\frac{1}{8}\lim_{\epsilon\to 0}\int_{\rho=\epsilon}\left(\partial_{\rho}\left(6 U+2\alpha\right)-2\alpha\partial_\rho\log\rho\right)\\
&-\frac{1}{8}\lim_{r\to\infty}\int_{S_r}\left(\partial_r\left[6U+2\alpha\right]-4r^{-1}\alpha \right)\,dS_r\\
%&=\frac{1}{8}\lim_{\epsilon\to 0}\int_{\rho=\epsilon}\left(\partial_{\rho}\left(6 U+2\alpha\right)-2\alpha\partial_\rho\log\rho\right)+m\\
%&=\frac{1}{8}\lim_{\epsilon\to 0}\int_{\rho=\epsilon}\left(\rho\partial_{\rho}\left(6 U+2\alpha\right)-2\alpha\right)dzd\phi+m\\
=&\frac{\pi}{4}\int_{H}\rho\partial_{\rho}\left(6 U+2\alpha\right) dz-\frac{\pi}{2}\int_{\Gamma}\alpha dz+m\\
=&-\frac{\pi}{2}\int_{-\mathfrak{m}}^{\mathfrak{m}}dz-\frac{\pi}{2}\int_{\Gamma}\alpha dz+m\\
=&-\pi \mathfrak{m}-\frac{\pi}{2}\int_{\Gamma}\alpha dz+m
\end{split}
\end{equation}Integrating the fifth term in the right hand of \eqref{SCALAR} yields 
\begin{align}
\begin{split}\label{boundary1}
\frac{1}{8}\int_{\mathbb{R}^{3}}\delta_{3}(\nabla h_2,\nabla V)\,dx
=&-\lim_{\varepsilon\rightarrow 0}\frac{1}{8}\int_{\rho=\varepsilon}V\partial_{\rho}h_2\\
=&\frac{\pi}{4}\int_{\Gamma_-}V\,dz-\frac{\pi}{4}\int_{\Gamma_+ }V\,dz+\frac{\pi}{4}\int_{H}V dz\\
\end{split}
\end{align}
Combining \eqref{SCALAR} and \eqref{mass1}, the total ADM mass is
\begin{align}\label{massf}
	\begin{split}
		m=&\frac{1}{16}\int_{\mathbb{R}^{3}}\left(12|\nabla U|^{2}+|\nabla V|^2+|\nabla W|^2+\sinh^{2}W\abs{\nabla\left(V+h_2\right)}^2\right)dx\\
		&+\frac{1}{16}\int_{\mathbb{R}^{3}}e^{2U+2\alpha-\log\left(2\sqrt{\rho^2+(z-\mathfrak{m})^2}\right)}R\,dx\\
		&+\frac{1}{32}\int_{\mathbb{R}^{3}}
		e^{-2\alpha+\log\left(2\sqrt{\rho^2+(z-\mathfrak{m})^2}\right)}\lambda_{ij}(A_{\rho,z}^{i}-A_{z,\rho}^{i})
		(A_{\rho,z}^{j}-A_{z,\rho}^{j})dx\\
		&+\frac{\pi}{2}\int_{\Gamma}\alpha dz+\pi \mathfrak{m},
	\end{split}
\end{align}
where the volume form $dx$ is again with respect to $\delta_{3}$. Next, we break up the data $(\alpha,U,V)$ in the following way
\begin{equation}
	U=U_S+\overline{U},\quad \alpha=\alpha_S+\overline{\alpha},\quad V=V_{S}+\overline{V}.
\end{equation}Combining this with \eqref{massf}, we obtain \begin{align}
	\begin{split}
		m=&\frac{1}{16}\int_{\mathbb{R}^{3}}12|\nabla \overline{U}|^{2}+12|\nabla U_S|^{2}+|\nabla V|^{2}+|\nabla V_S|^{2}+|\nabla W|^{2}
		+\sinh^{2}W|\nabla (V+h_2)|^{2}
		dx\\
		&+\frac{1}{16}\int_{\mathbb{R}^{3}}\frac{1}{4}e^{-2\alpha+\log\left(2\sqrt{\rho^2+(z-\mathfrak{m})^2}\right)}\lambda_{ij}(A_{\rho,z}^{i}-A_{z,\rho}^{i})
		(A_{\rho,z}^{j}-A_{z,\rho}^{j})dx\\
		&+\frac{1}{8}\int_{\mathbb{R}^{3}}
		e^{2U+2\alpha-\log\left(2\sqrt{\rho^2+(z-\mathfrak{m})^2}\right)}Rdx-\frac{1}{8}\int_{\R^3}\left(12\overline{U}\Delta U_S+\overline{V}\Delta V_{S}\right) dx\\
		&+\frac{1}{8}\lim_{r\to\infty}\int_{S_r}\left(12\overline{U}\partial_r U_S+\overline{V}\partial_r V_{S}\right) -\frac{1}{8}\lim_{\epsilon\to 0}\int_{\rho=\epsilon}\left(12\overline{U}\partial_{\rho} U_S+\overline{V}\partial_{\rho} V_{S}\right)\\
		&-\frac{1}{8}\lim_{\epsilon\to 0}\int_{\rho=\epsilon}\rho\partial_{\rho}\left(6 U+2\alpha\right)+\frac{\pi}{4}\int_{H}(2\alpha+V) dz+\pi \mathfrak{m} .
	\end{split}
\end{align}
Since $U,V=O(1)$ on $\Gamma_-\cup \Gamma_+$, the integrals on axis are non-zero only on horizon rod. Since $V_{S}$ and $U_S$ are harmonic on $\R^3 \setminus \Gamma$ , the forth integral vanishes. Moreover, because of the estimates of $\overline{U},U_S,\overline{V}$, and $V_{S}$ at infinity, the fifth integral vanishes. As in the computation leading to \eqref{mass1}, the sixth integral is $\pi m_S$ and  $\rho\partial_{\rho} U_S=-\frac{1}{2},\rho\partial_{\rho}V_{S}=0$ as $\rho\to 0$. Combining this with the Schwarzschild mass formula \eqref{mS}, we have 
\begin{align}\label{m1}
	\begin{split}
		m=&\frac{1}{16}\int_{\mathbb{R}^{3}}12|\nabla \overline{U}|^{2}+|\nabla \overline{V}|^{2}+|\nabla W|^{2}
		+\sinh^{2}W|\nabla (V+h_2)|^{2}
		dx\\
		&+\frac{1}{16}\int_{\mathbb{R}^{3}}\frac{1}{4}e^{-2\alpha+\log\left(2\sqrt{\rho^2+(z-\mathfrak{m})^2}\right)}\lambda_{ij}(A_{\rho,z}^{i}-A_{z,\rho}^{i})
		(A_{\rho,z}^{j}-A_{z,\rho}^{j})dx\\
		&+\frac{1}{8}\int_{\mathbb{R}^{3}}
		e^{2U+2\alpha-\log\left(2\sqrt{\rho^2+(z-\mathfrak{m})^2}\right)}Rdx+\frac{\pi}{4}\int_{H}\left(2\overline{\alpha}+6\overline{U}+\overline{V}\right)dz+m_S
	\end{split}
\end{align}
where $m_S$ is the mass of the Schwarzschild initial data \eqref{mS}. Next, we follow \cite[Section 3]{Alaeemassam}. Under the assumptions that $J(\eta_{(l)})=0$, $l=1,2$ and $M^{4}$ is simply connected, twist potentials $\zeta^i$ exist such that
\begin{equation}
d\zeta^{l}=2\star \left(p(\eta_{(l)})\wedge\eta_{(1)}\wedge\eta_{(2)}\right), \text{ }\text{ }\text{ }\text{ }\text{ }\text{ }l=1,2.
\end{equation}where $\star$ is the Hodge star with respect to the metric $g$. Moreover, by the constraint equation and $\mu\geq 0$, we have 
\begin{equation}\label{R1}
\begin{split}
e^{2U+2\alpha-\log\left(2\sqrt{\rho^2+z^2}\right)}R&= e^{2U+2\alpha-\log\left(2\sqrt{\rho^2+z^2}\right)}\left(16\pi\mu+|k|_g\right)\\
&\geq \frac{1}{2}e^{-6h_{1}-6U+h_{2}+V}\cosh W
\left|e^{-h_{2}-V}\tanh W\nabla \zeta^{1}-\nabla \zeta^{2}\right|^{2}\\
&+\frac{1}{2}\frac{e^{-6h_{1}-6U-h_{2}-V}}{\cosh W}|\nabla \zeta^{1}|^{2}
\end{split}
\end{equation}Combining \eqref{m1} and \eqref{R1}, we have 
\begin{align}\label{mass}
\begin{split}
m\geq &\mathcal{I}_{\Omega}(\Psi)+\frac{\pi}{4}\int_{H}\left(2\overline{\alpha}+6\overline{U}-\text{sgn}(z)\overline{V}\right)dz+m_S
\end{split}
\end{align}where $\Psi=(\overline{U},\overline{V},{W},\zeta^{1},\zeta^{2})$  and
\begin{align}
\begin{split}
\mathcal{I}_{\Omega}(\Psi)=&\int_{\Omega}12|\nabla \overline{U}|^{2}+|\nabla \overline{V}|^{2}+|\nabla W|^{2}
+\sinh^{2}W|\nabla (V+h_{2})|^{2}
+\frac{e^{-6h_{1}-6U-h_{2}-V}}{\cosh W}|\nabla \zeta^{1}|^{2}dx
\\
&+\int_{\Omega}
e^{-6h_{1}-6U+h_{2}+V}\cosh W
\left|e^{-h_{2}-V}\tanh W\nabla \zeta^{1}-\nabla \zeta^{2}\right|^{2}dx.
\end{split}
\end{align}The right hand side of the above inequality is related to the harmonic energy of a map $\tilde{\Psi}=(u=U+h_1,v=V+h_2,w=W,\zeta^{1},\zeta^{2}):\mathbb{R}^{3}\,\backslash\,\Gamma_-\cup \Gamma_+\rightarrow
SL(3,\mathbb{R})/SO(3)$, on a domain $\Omega\subset\mathbb{R}^{3}$ \cite{Alaeemassam} where the target space metric is an Einstein metric with negative curvature.  The harmonic energy is given by
\begin{align}\label{energy1}
\begin{split}
E_{\Omega}(\tilde{\Psi})=&\int_{\Omega}12|\nabla u|^{2}
+\cosh^{2}w|\nabla v|^{2}
+|\nabla w|^{2}
+\frac{e^{-6u-v}}{\cosh w}|\nabla \zeta^{1}|^{2}dx
\\
&+\int_{\Omega}
e^{-6u+v}\cosh w
\left|e^{-v}\tanh w\nabla \zeta^{1}-\nabla \zeta^{2}\right|^{2}dx.
\end{split}
\end{align}
This energy is related to the reduced energy $\mathcal{I}_{\Omega}$ of the map ${\Psi}=(\overline{U},\overline{V},{W},\zeta^{1},\zeta^{2})$ may be expressed in terms of the harmonic energy of $\tilde{\Psi}$ by
\begin{equation}\label{51}
\mathcal{I}_{\Omega}(\Psi)=E_{\Omega}(\tilde{\Psi})
-12\int_{\partial\Omega}
(2u-h_1-U_S)\partial_{\nu}\left(h_1+U_S\right)-\int_{\partial\Omega}(2v-h_2-V_S)\partial_{\nu}\left(h_{2}+V_S\right),
\end{equation}
where $\nu$ denotes the unit outer normal to the boundary $\partial\Omega$. 
Let $\tilde{\Psi}_{MP}=(u_{MP},v_{MP},w_{MP},\zeta^{1}_{MP},\zeta^{2}_{MP})$ denote the Myers-Perry harmonic map  with the same J and A (see Appendix B), and
let $\Psi_{MP}$ be the associated renormalized map with $u_{MP}=\overline{U}_{MP}+U_S+h_{1}$, $v_{MP}=\overline{V}_{MP}+V_S+h_{2}$, and $w_{MP}=W_{MP}$. Therefore,
$\Psi_{MP}$ is a critical point of $\mathcal{I}$. We will show that $\Psi_{MP}$ achieves the global minimum for $\mathcal{I}$.

\begin{theorem}\label{infimum}
	Suppose that $\Psi=(\overline{U},\overline{V},W,\zeta^{1},\zeta^{2})$ is smooth and satisfies the asymptotics \eqref{fall1}-\eqref{fall4.1}
	with $\zeta^{1}|_{\Gamma}=\zeta^{1}_{MP}|_{\Gamma}$ and $\zeta^{2}|_{\Gamma}=\zeta^{2}_{MP}|_{\Gamma}$, then there exists a constant $C>0$ such that
	\begin{equation}\label{53}
	\mathcal{I}(\Psi)-\mathcal{I}(\Psi_{MP})
	\geq C\left(\int_{\mathbb{R}^{3}}
	\operatorname{dist}_{SL(3,\mathbb{R})/SO(3)}^{6}(\Psi,\Psi_{MP})dx
	\right)^{\frac{1}{3}}.
	\end{equation}
\end{theorem} In the following sections we establish this result. 

\section{asymptotics in Weyl coordinate}\label{sec:asymp}
In this section, we list all asymptotic behaviour of a map $\Psi$ that ensures a finite reduced energy $\mathcal{I}(\Psi)$. Inspired from the Myers-Perry map, we require that 
\begin{equation}
\overline{U},\overline{V}\in C^{0,1}(\mathbb{R}^3),\quad\text{and}\qquad \overline{U},\overline{V}=O_1(r^{-1-\kappa}),\quad \text{as}\quad r\to\infty
\end{equation} for $\kappa>0$. We require that as $r\rightarrow\infty$ the following decay occur
\begin{equation}\label{fall0}
h_1=\frac{1}{2}\log\rho+O(r^{-2}),\text{ }\text{ }\text{ }\text{ }h_2=\log\left(\tan\theta\right)+O(r^{-2}),
\end{equation}
\begin{equation}\label{fall1}
\overline{U}=O(r^{-1-\kappa}),\text{ }\text{ }\text{ }\text{ }	\overline{V}=O(r^{-1-\kappa}),\text{ }\text{ }\text{ }\text{ }W=\sqrt{\rho} O(r^{-2-\kappa}),
\end{equation}
\begin{equation}\label{fall2}
|\nabla	\overline{U}|=O(r^{-3-\kappa}),\text{ }\text{ }\text{ }\text{ }|\nabla	\overline{V}|=O(r^{-3-\kappa}),\text{ }\text{ }\text{ }\text{ }|\nabla W|=\rho^{-\frac{1}{2}} O(r^{-2-\kappa}),
\end{equation}
\begin{equation}\label{fall3}
|\nabla \zeta^{1}|=\rho\sqrt{\sin\theta} O(r^{-2-\kappa}),\text{ }\text{ }\text{ }\text{ }\text{ }\text{ }\text{ }\text{ }|\nabla \zeta^{2}|=\rho\sqrt{\cos\theta} O(r^{-2-\kappa}).
\end{equation}
Next consider asymptotic behaviour near  the rod points $p_{\pm}$. As $r_+\rightarrow 0$, we require 
\begin{equation}\label{l1}
h_1=\frac{1}{2}\log \rho,\qquad h_2=\frac{1}{2}\log\left(\frac{r_+-(z-\mathfrak{m})}{r_++(z-\mathfrak{m})}\right)
\end{equation}
\begin{equation}\label{fall1.0}
\overline{U}=O(1),\text{ }\text{ }\text{ }\text{ }\overline{V}=O(1),\text{ }\text{ }\text{ }\text{ } W=O(1),
\end{equation}
\begin{equation}\label{fall2.0}
|\nabla\overline{U}|=O(r_+^{-1}),\text{ }\text{ }\text{ }\text{ }|\nabla\overline{V}|=O(r_+^{-1}),\text{ }\text{ }\text{ }\text{ }|\nabla W|=O(r_+^{-1}),
\end{equation}
\begin{equation}\label{fall4.0}
|\nabla \zeta^{1}|=\rho^{3/2} O(r_+^{-3/2}),\quad \text{for}\quad z\geq \mathfrak{m},\qquad |\nabla \zeta^{1}|=O(1),\quad\text{for}\quad z\leq \mathfrak{m}
\end{equation}
\begin{equation}\label{fall4.31}
|\nabla \zeta^{2}|=\rho^{1/2} O(r_+^{-1/2}),\quad \text{for}\quad z\geq \mathfrak{m},\qquad |\nabla \zeta^{2}|=O(1),\quad\text{for}\quad z\leq \mathfrak{m}
\end{equation} As $r_-\rightarrow 0$, we require 
\begin{equation}
h_1=\frac{1}{2}\log \rho,\qquad h_2=\frac{1}{2}\log\left(\frac{r_+-(z-\mathfrak{m})}{r_++(z-\mathfrak{m})}\right)
\end{equation}
\begin{equation}\label{fall1.01}
\overline{U}=O(1),\text{ }\text{ }\text{ }\text{ }\overline{V}=O(1),\text{ }\text{ }\text{ }\text{ } W=O(1),
\end{equation}
\begin{equation}\label{fall2.01}
|\nabla\overline{U}|=O(r_-^{-1}),\text{ }\text{ }\text{ }\text{ }|\nabla\overline{V}|=O(r_-^{-1}),\text{ }\text{ }\text{ }\text{ }|\nabla W|=O(r_-^{-1}),
\end{equation}
\begin{equation}\label{fall4.01}
|\nabla \zeta^{1}|=O(1),\quad \text{for}\quad z\leq-\mathfrak{m},\qquad |\nabla \zeta^{1}|=\rho^{3/2} O(r_-^{-3/2}),\quad\text{for}\quad z\geq- \mathfrak{m}
\end{equation}
\begin{equation}\label{fall4.3}
|\nabla \zeta^{2}|=O(1),\quad \text{for}\quad z\leq- \mathfrak{m},\qquad |\nabla \zeta^{2}|=\rho^{1/2} O(r_-^{-1/2}),\quad\text{for}\quad z\geq -\mathfrak{m}
\end{equation}By integrating \eqref{fall4.0} on lines perpendicular to the axis at $p_\pm$ and using the fact that $(\zeta^{i}-\zeta_{MP}^{i})|_{\Gamma}=0$, we have 
\begin{equation}\label{z1}
\zeta^1-\zeta^1_{MP}=O(\rho^{5/2}r_{\pm}^{-3/2}),\qquad \quad \text{as $r_\pm\to 0$ and $|z|\geq \mathfrak{m}$}
\end{equation} 
\begin{equation}\label{z2}
\zeta^2-\zeta^2_{MP}=O(\rho^{3/2}r_{\pm}^{-1/2}),\qquad \quad \text{as $r_\pm\to 0$ and $|z|\geq \mathfrak{m}$}
\end{equation}
On horizon rod $|z|\leq\mathfrak{m}$, the potentials do not match and we integrate along the radial line emanating from the poles $p_\pm$ and we have
\begin{equation}\label{zih}
\zeta^i-\zeta^i_{MP}=O(r_{\pm}),\qquad \quad \text{as $r_\pm\to 0$ and $|z|\leq \mathfrak{m}$}
\end{equation}

Next, the asymptotics on approach to the symmetry axes as $\rho\rightarrow 0$, we have 
\begin{equation}\label{fall0.1}
h_1=\frac{1}{2}\log\rho,\qquad \text{and}\quad h_2=\text{sgn}(z-\mathfrak{m})\log\rho+O(1)
\end{equation}Furthermore, for $|z|\leq \mathfrak{m}$, the functions are required to satisfy
\begin{equation}\label{fall1.1}
\overline{U}=O(1),\text{ }\text{ }\text{ }\text{ }{}\overline{V}=O(1),\text{ }\text{ }\text{ }\text{ }W=O(1),
\end{equation}
\begin{equation}\label{fall2.1}
|\nabla\overline{U}|=O(1),\text{ }\text{ }\text{ }\text{ }|\nabla\overline{V}|=O(1),\text{ }\text{ }\text{ }\text{ }|\nabla W|=O(1),
\end{equation}
\begin{equation}\label{fall4.1}
|\nabla \zeta^{1}|=O(1),\text{ }\text{ }\text{ }\text{ }\text{ }\text{ }\text{ }\text{ }|\nabla \zeta^{2}|=O(1).
\end{equation}
and for $|z|>\mathfrak{m}$, we have 
\begin{equation}\label{fall1.2}
\overline{U}=O(1),\text{ }\text{ }\text{ }\text{ }\overline{V}=O(1),\text{ }\text{ }\text{ }\text{ }W=O(\rho^{\frac{1}{2}}),
\end{equation}
\begin{equation}\label{fall2.2}
|\nabla\overline{V}|=O(1),\text{ }\text{ }\text{ }\text{ }|\nabla\overline{V}|=O(1),\text{ }\text{ }\text{ }\text{ }|\nabla W|=O(\rho^{-\frac{1}{2}}),
\end{equation}
\begin{equation}\label{fall4.2}
|\nabla \zeta^{1}|=\sqrt{\sin\theta}O(\rho),\text{ }\text{ }\text{ }\text{ }\text{ }\text{ }\text{ }\text{ }|\nabla \zeta^{2}|=\sqrt{\cos\theta}O(\rho).
\end{equation}By integrating \eqref{fall4.0} on lines perpendicular to the axis and with the fact that $(\zeta^{i}-\zeta_{MP}^{i})|_{\Gamma}=0$, we have
\begin{equation}
\zeta^i=\text{constant}+O(\rho^2),\qquad \text{for $|z|>\mathfrak{m}$}.
\end{equation}

\section{The Cut-and-Paste Argument}
We now prove Theorem \ref{infimum} using the convexity of the harmonic energy $E$ for maps $\Psi$ for nonpositively curved target space metric under geodesic deformations. This requires a detailed analysis of the reduced energy $I(\Psi)$ on different regions for singular harmonic maps using a cut-and-paste argument. First, we approximate the map $\Psi$ with an associated map $\Psi_{\delta,\varepsilon}$ which agree with $\Psi_{MP}$ on certain regions. In particular, let $\delta,\varepsilon>0$ be
small parameters and define sets $\Omega_{\delta,\varepsilon}=\{\delta<r_\pm;\,r<2/\delta;
\rho>\varepsilon\}$ and $\mathcal{A}_{\delta,\varepsilon}=B_{2/\delta}\setminus
\Omega_{\delta,\varepsilon}$, where $B_{2/\delta}$ is the ball of radius $2/\delta$ centered at the origin. Then, the approximate map $\Psi_{\delta,\varepsilon}=(\overline{U}_{\delta,\varepsilon},\overline{V}_{\delta,\varepsilon},W_{\delta,\varepsilon},\zeta^i_{\delta,\varepsilon})$ must satisfy the following properties
\begin{equation}\label{54}
	\operatorname{supp}(\overline{U}_{\delta,\varepsilon}-\overline{U}_{MP})\subset B_{2/\delta},\text{ }\text{ }\text{ }\text{ }\text{ }
	\operatorname{supp}(\overline{V}_{\delta,\varepsilon}-\overline{V}_{MP},W_{\delta,\varepsilon}-W_{MP},
	\zeta^{1}_{\delta,\varepsilon}-\zeta^{1}_{MP},\zeta^{2}_{\delta,\varepsilon}-\zeta^{2}_{MP})\subset \Omega_{\delta,\varepsilon}.
\end{equation}
Let $\tilde{\Psi}^{t}_{\delta,\varepsilon}$, $t\in[0,1]$, be a geodesic in $SL(3,\mathbb{R})/SO(3)$ which connects
$\tilde{\Psi}^{1}_{\delta,\varepsilon}=\tilde{\Psi}_{\delta,\varepsilon}$ and $\tilde{\Psi}^{0}_{\delta,\varepsilon}=\tilde{\Psi}_{MP}$. Then the properties \eqref{54} imply that $\tilde{\Psi}^{t}_{\delta,\varepsilon}\equiv\tilde{\Psi}_{MP}$ outside $B_{2/\delta}$ and
$(\overline{V}^{t}_{\delta,\varepsilon},{W}^{t}_{\delta,\varepsilon},\zeta^{i,t}_{\delta,\varepsilon})=(\overline{V}_{MP},{W}_{MP},\zeta^{i}_{MP})$ for $i=1,2$,
in a neighborhood of $\mathcal{A}_{\delta,\varepsilon}$. Furthermore, we impose that $\overline{U}^{t}_{\delta,\varepsilon}=\overline{U}_{MP}+t(\overline{U}_{\delta,\varepsilon}-\overline{U}_{MP})$ and $\overline{V}^{t}=\overline{V}_{MP}$ on these
regions. 

Combining the convexity of the harmonic energy $E$, the linear behaviour of $\overline{U}^{t}_{\delta,\varepsilon}$ in $t$, and constancy of $\overline{V}^{t}_{\delta,\varepsilon}$, we obtain the following inequality,  similar to that found  in \cite{Alaeemassam}:
\begin{equation}
\frac{d^{2}}{dt^{2}}\mathcal{I}({\Psi^{t}}_{\delta,\varepsilon})
\geq 2\int_{\mathbb{R}^{3}}|\nabla\operatorname{dist}_{SL(3,\mathbb{R})/SO(3)}
({\Psi}_{\delta,\varepsilon},{\Psi}_{MP})|^{2}dx
\end{equation}Integrating this and using the fact that the map $\Psi_M$ is a critical point of the reduced energy and applying a Sobolev inequality, we obtain the gap bound in Theorem \ref{infimum} as $\delta,\varepsilon\to 0$. 

To construct the approximate map $\Psi_{\delta,\varepsilon}$ introduced above,  we define smooth cut-off functions, which only take values in the interval $[0,1]$, by
\begin{equation}\label{65}
\overline{\varphi}_{\delta}=\begin{cases}
1 & \text{ if $r\leq\frac{1}{\delta}$,} \\
|\nabla\overline{\varphi}_{\delta}|\leq 2\delta^2 &
\text{ if $\frac{1}{\delta}<r<\frac{2}{\delta}$,} \\
0 & \text{ if $r\geq\frac{2}{\delta}$,} \\
\end{cases}
\end{equation}
\begin{equation}\label{66}
\varphi_{\delta}=\begin{cases}
0 & \text{ if $r_\pm\leq\delta$,} \\
|\nabla\varphi_{\delta}|\leq \frac{2}{\delta} &
\text{ if $\delta<r_\pm<2\delta$,} \\
1 & \text{ if $r_\pm\geq2\delta$,} \\
\end{cases}
\end{equation}
and
\begin{equation}\label{67}
\phi_{\varepsilon}=\begin{cases}
0 & \text{ if $\rho\leq\varepsilon$,} \\
\frac{\log(\rho/\varepsilon)}{\log(\sqrt{\varepsilon}/
	\varepsilon)} &
\text{ if $\varepsilon<\rho<\sqrt{\varepsilon}$,} \\
1 & \text{ if $\rho\geq\sqrt{\varepsilon}$.} \\
\end{cases}
\end{equation}
First, we deal with the asymptotically flat region. Let
\begin{equation}\label{681}
	\overline{F}_{\delta}({\Psi})={\Psi}_{MP}
	+\overline{\varphi}_{\delta}({\Psi}-{\Psi}_{MP})
\end{equation}
so that $\overline{F}_{\delta}({\Psi})={\Psi}_{MP}$ on $\mathbb{R}^{3}\setminus B_{2/\delta}$. 

\begin{lemma}\label{cutandpaste1}
	$\lim_{\delta\rightarrow 0}\mathcal{I}(\overline{F}_{\delta}({\Psi}))=\mathcal{I}({\Psi}).$
\end{lemma} The proof follows from \cite[Lemma 4.2]{Alaeemassam}. The next regions are neighbourhoods of the rod points $p_\pm$. Let
\begin{equation}\label{68}
	F_{\delta}({\Psi})=
	(\overline{U},{V}_{\delta},W_{\delta}
	,\zeta^{1}_{\delta},\zeta^{2}_{\delta}),
\end{equation}
where
\begin{equation}
	(\overline{V}_{\delta},W_{\delta}
	,\zeta^{1}_{\delta},\zeta^{2}_{\delta})=(\overline{V}_{MP},W_{MP},\zeta^{1}_{MP},\zeta^{2}_{MP})
	+\varphi_{\delta}(\overline{V}-\overline{V}_{MP},W-W_{MP},\zeta^{1}-\zeta^{1}_{MP},\zeta^{2}-\zeta^{2}_{MP}),
\end{equation}
so that $F_{\delta}({\Psi})={\Psi}_{MP}$ on $B_{\delta}(p_+)\cup B_{\delta}(p_-)$.

\begin{lemma}\label{cutandpaste2}
	$\lim_{\delta\rightarrow 0}\mathcal{I}(F_{\delta}({\Psi}))=\mathcal{I}({\Psi}).$ This also holds if ${\Psi}\equiv{\Psi}_{MP}$ outside of $B_{2/\delta}$.
\end{lemma}

\begin{proof}
	Write
	\begin{equation}\label{080}
		\mathcal{I}(F_{\delta}({\Psi}))
		=\sum_{\pm}\left[\mathcal{I}_{r_\pm\leq\delta}(F_{\delta}({\Psi}))
		+\mathcal{I}_{\delta< r_\pm<2\delta}(F_{\delta}({\Psi}))\right]
		+\mathcal{I}_{r_\pm\geq 2\delta}(F_{\delta}({\Psi})),
	\end{equation}where $r_\pm\geq 2\delta$ is outside the open balls $B_{2\delta}(p_+)\cup B_{2\delta}(p_-)$. Observe that by the dominated convergence theorem (DCT)
	\begin{equation}\label{081}
		\mathcal{I}_{r_\pm\geq2\delta}(F_{\delta}({\Psi}))
		=\mathcal{I}_{r_\pm\geq2\delta}({\Psi})
		\rightarrow \mathcal{I}({\Psi})
	\end{equation}
	Moreover
	\begin{align}
		\begin{split}
			\mathcal{I}_{r_\pm\leq\delta}({\Psi})&=\int_{r_\pm\leq\delta}12|\nabla\overline{U}|^{2}
			+|\nabla \overline{V}_{MP}|^{2}+|\nabla W_{MP}|^{2}
			\\
			&+\int_{r_\pm\leq\delta}\sinh^{2}W_{MP}|\nabla({V}_{MP}+h_2)|^{2}+\frac{e^{-6h_1-6{U}-h_2-{V}_{MP}}}{\cosh W_{MP}}|\nabla \zeta_{MP}^{1}|^{2}\\
			&+\int_{r_\pm\leq\delta}
			e^{-6h_1-6{U}+h_2+{V}_{MP}}\cosh W_{MP}
			\left|\nabla \zeta_{MP}^{2}-e^{-h_2-{V}_{MP}}\tanh W_{MP}\nabla \zeta_{MP}^{1}\right|^{2},
		\end{split}
	\end{align}
	where the first two term in the first line converge to zero again by the DCT. The remaining terms converge to zero by the reduced energy of $\Psi_{MP}$ and $e^{-{U}}\leq Ce^{-{U}_{MP}}$ and $e^{{V}}\leq Ce^{{V}_{MP}}$ near the rod points for some positive constant $C$.
	
	Now consider
	\begin{align}\label{070}
		\begin{split}
			\mathcal{I}_{\delta< r_\pm<2\delta}(F_{\delta}({\Psi}))
			=&\underbrace{\int_{\delta< r_\pm<2\delta}
				12|\nabla \overline{U}|^{2}}_{I_{1}}
			+\underbrace{\int_{\delta< r_\pm<2\delta}
				|\nabla V_{\delta}|^{2}+\int_{\delta< r_\pm<2\delta}
				|\nabla W_{\delta}|^{2}}_{I_{2}}\\
			&+\underbrace{\int_{\delta< r_\pm<2\delta}
				\sinh^{2}W_{\delta}|\nabla(V_{\delta}+h_2)|^{2}}_{I_{3}}
			+\underbrace{\int_{\delta< r_\pm<2\delta}
				\frac{e^{-6h_{1}-6{U}-h_2-{V}_{\delta}}}{\cosh W_{\delta}}
				|\nabla \zeta^{1}_{\delta}|^{2}}_{I_{4}}\\
			&+\underbrace{\int_{\delta< r_\pm<2\delta}
				e^{-6h_1-6{U}+h_2+{V}_{\delta}}\cosh W_{\delta}
				|\nabla \zeta^{2}_{\delta}-e^{-V_{\delta}}\cot\theta \tanh W_{\delta}
				\nabla \zeta^{1}_{\delta}|^{2}}_{I_{5}}.
		\end{split}
	\end{align}where $V_{\delta}=\overline{V}_{\delta}+V_S$. By the DCT, $I_1\to 0$. By expanding terms in $V_{\delta}$ and $W_{\delta}$, we have 
	\begin{equation}\label{00072}
		I_{2}\leq C\int_{\delta< r_\pm<2\delta}
		\left(|\nabla \overline{V}|^{2}+|\nabla \overline{V}_{MP}|^{2}+|\nabla W|^{2}+|\nabla W_{MP}|^{2}
		+\frac{1}{r_\pm^2}\underbrace{(\overline{V}-\overline{V}_{MP})^{2}}_{O(1)}+\frac{1}{r_\pm^2}\underbrace{(W-W_{MP})^{2}}_{O(1)}\right),
	\end{equation}The first four terms converge to zero by the DCT and finite energy of $\Psi$ and $\Psi_{MP}$. The last two terms converge to zero by the DCT.  Using $\sinh^{2}W_{\delta}=O(1)$ and expand terms in $I_{3}$, we obtain
	\begin{equation}
		I_{3}\leq CI_2+C\int_{\delta< r_\pm<2\delta}
		\sinh^{2}W_{MP}|\nabla ({V}_{MP}+h_2)|^{2}
	\end{equation}where ${V}_{MP}=\overline{V}_{MP}+V_S$. Since the last term is bounded by the finite energy of $\Psi_{MP}$ and $I_2\to0$, the DCT implies that $I_3\to 0$. For $I_4$, we have the following estimates. 
	\begin{equation}
		I_{4}\leq C\int_{\delta< r_\pm<2\delta} \left(\frac{e^{-6h_1-6{U}-h_2-{V}}}{\cosh W}|\nabla \zeta^{1}|^{2}+
		\frac{e^{-6h_1-6{U}_{MP}-h_2-{V}_{MP}}}{\cosh W_{MP}}|\nabla \zeta^{1}_{MP}|^{2}
		+\frac{e^{-6h_1-6{U}_{MP}-h_2-{V}_{MP}}}{r_{\pm}^2\cosh W_{MP}}(\zeta^{1}-\zeta^{1}_{MP})^{2}\right)
	\end{equation} As above, the first two terms are bounded by finite energy of maps $\Psi$ and $\Psi_{MP}$ and they converge to zero by the DCT. The last term also converges to zero using equation \eqref{z1} and \eqref{zih} and the DCT.  Similar reasoning implies $I_5\to 0$.
\end{proof}

Consider now the cylindrical regions around the axis $\Gamma$ and away from the origin given by
\begin{equation}\label{91}
	\mathcal{C}_{\delta,\varepsilon}=
	\{\rho\leq\varepsilon\}\cap
	\{\delta\leq r_\pm\,;r\leq 2/\delta\},
\end{equation}
\begin{equation}\label{91.1}
	\mathcal{W}^1_{\delta,\varepsilon}=
	\{\varepsilon\leq \rho\leq\sqrt{\varepsilon}\}\cap
	\{\delta\leq r_\pm\,;r\leq 2/\delta\,;|z|>\mathfrak{m}\},
\end{equation} and
\begin{equation}\label{91.2}
\mathcal{W}^2_{\delta,\varepsilon}=
\{\varepsilon\leq \rho\leq\sqrt{\varepsilon}\}\cap
\{\delta\leq r_\pm\,;r\leq 2/\delta\,;|z|<\mathfrak{m}\}.
\end{equation}
Let
\begin{equation}\label{92}
	G_{\varepsilon}(\overline{\Psi})=(\overline{U},V_{\varepsilon},
	W_{\varepsilon},\zeta^{1}_{\varepsilon},\zeta^{2}_{\varepsilon})
\end{equation}
where
\begin{equation}\label{93}
	(V_{\varepsilon},
	W_{\varepsilon},\zeta^{1}_{\varepsilon},\zeta^{2}_{\varepsilon})=
	(\overline{V}_{MP},W_{MP},\zeta^{1}_{MP},\zeta^{2}_{MP})
	+\phi_{\varepsilon}(\overline{V}-\overline{V}_{MP},W-W_{MP},\zeta^{1}-\zeta^{1}_{MP},\zeta^{2}-\zeta^{2}_{MP}),
\end{equation}
so that $G_{\varepsilon}(\overline{\Psi})=\overline{\Psi}_{0}$ on $\rho\leq\varepsilon$.

\begin{lemma}\label{cutandpaste3}
	Fix $\delta>0$ and suppose that ${\Psi}\equiv{\Psi}_{MP}$ on $B_{\delta}(p_+)\cup B_{\delta}(p_-)$. Then
	$\lim_{\varepsilon\rightarrow 0}\mathcal{I}(G_{\varepsilon}({\Psi}))=\mathcal{I}({\Psi})$. This also holds if
	${\Psi}\equiv\overline{\Psi}_{MP}$ outside $B_{2/\delta}$.
\end{lemma}

\begin{proof}
	Write
	\begin{equation}\label{94}
		\mathcal{I}(G_{\varepsilon}({\Psi}))
		=\mathcal{I}_{\mathcal{C}_{\delta,\varepsilon}}(G_{\varepsilon}({\Psi}))
		+\mathcal{I}_{\mathcal{W}^1_{\delta,\varepsilon}}(G_{\varepsilon}({\Psi}))
		+\mathcal{I}_{\mathcal{W}^2_{\delta,\varepsilon}}(G_{\varepsilon}({\Psi}))
		+\mathcal{I}_{\mathbb{R}^{3}\setminus(\mathcal{C}_{\delta,\varepsilon}
			\cup\mathcal{W}^1_{\delta,\varepsilon})\cup\mathcal{W}^2_{\delta,\varepsilon})}(G_{\varepsilon}({\Psi})).
	\end{equation}
	Since
	${\Psi}\equiv{\Psi}_{MP}$ on $B_{\delta}(p_+)\cup B_{\delta}(p_-)$, the DCT and finite energy of ${\Psi}_{MP}$ imply that
	\begin{equation}\label{95}
		\mathcal{I}_{\mathbb{R}^{3}\setminus(\mathcal{C}_{\delta,\varepsilon}
			\cup\mathcal{W}^1_{\delta,\varepsilon})\cup\mathcal{W}^2_{\delta,\varepsilon})}(G_{\varepsilon}({\Psi}))
		\rightarrow \mathcal{I}({\Psi}).
	\end{equation}
	Moreover, we have 
	\begin{align}
		\begin{split}
			\mathcal{I}_{\mathcal{C}_{\delta,\varepsilon}}(G_{\varepsilon}({\Psi}))=&
			\int_{\mathcal{C}_{\delta,\varepsilon}}12|\nabla \overline{U}|^{2}
			+|\nabla \overline{V}_{MP}|^{2}+|\nabla W_{MP}|^{2}
			+\sinh^{2}W_{MP}|\nabla({V}_{MP}+h_2)|^{2}+\frac{e^{-6h_1-6{U}-h_2-{V}_{MP}}}{\cosh W_{MP}}|\nabla \zeta_{MP}^{1}|^{2}
			\\
			&+\int_{\mathcal{C}_{\delta,\varepsilon}}
			e^{-6h_1-6{U}+h_2+{V}_{MP}}\cosh W_{MP}
			\left|\nabla \zeta_{MP}^{2}-e^{-h_2-{V}_{MP}}\tanh W_{MP}\nabla \zeta_{MP}^{1}\right|^{2}.
		\end{split}
	\end{align}The first term converges to zero because of boundedness $|\nabla \overline{U}|$ and the DCT. Combining the fact that the potentials of $G_{\varepsilon}({\Psi})$ agrees with $\Psi_{MP}$ on $\mathcal{C}_{\delta,\varepsilon}$, $e^{U}\leq C e^{U_{MP}}$, $e^{V}\leq C e^{V_{MP}}$,  the energy of $\Psi$ and $\overline{\Psi}$, and the DCT, all other terms converge to zero.
	
	The term $\mathcal{I}_{\mathcal{W}^1_{\delta,\varepsilon}}(G_{\varepsilon}({\Psi}))$ converges to zero from \cite[Lemma 4.4.]{Alaeemassam}. Now consider the last term
	\begin{align}
		\begin{split}
			\mathcal{I}_{\mathcal{W}^2_{\delta,\varepsilon}}(G_{\varepsilon}({\Psi}))
			=&\underbrace{\int_{\mathcal{W}^2_{\delta,\varepsilon}}
				12|\nabla \overline{U}|^{2}}_{I_{1}}
			+\underbrace{\int_{\mathcal{W}^2_{\delta,\varepsilon}}
				|\nabla {V}_{\varepsilon}|^{2}
			+\int_{\mathcal{W}^2_{\delta,\varepsilon}}
				|\nabla W_{\varepsilon}|^{2}}_{I_{2}}\\
			&+\underbrace{\int_{\mathcal{W}^2_{\delta,\varepsilon}}
				\sinh^{2}W_{\varepsilon}|\nabla({V}_{\varepsilon}+h_{2})|^{2}}_{I_{3}}
			+\underbrace{\int_{\mathcal{W}^2_{\delta,\varepsilon}}
				\frac{e^{-6h_1-6{U}-h_2-{V}_{\varepsilon}}}{\cosh W_{\varepsilon}}
				|\nabla \zeta^{1}_{\delta}|^{2}}_{I_{4}}\\
			&+\underbrace{\int_{\mathcal{W}^2_{\delta,\varepsilon}}
				e^{-6h_1-6{U}+h_2+{V}_{\varepsilon}}\cosh W_{\varepsilon}
				|\nabla \zeta^{2}_{\varepsilon}-e^{-h_2^S-V_{\varepsilon}}\tanh W_{\varepsilon}
				\nabla \zeta^{1}_{\varepsilon}|^{2}}_{I_{5}}.
		\end{split}
	\end{align}The first term $I_1\to 0$ using boundedness of $|\nabla \overline{U}|^{2}$. 
	We have
\begin{equation}
I_{2}\leq C\int_{\mathcal{W}^2_{\delta,\varepsilon}}
\left(|\nabla \overline{V}|^{2}+|\nabla \overline{V}_{MP}|^{2}+|\nabla W|^{2}+|\nabla W_{MP}|^{2}
+(\log\varepsilon)^{-2}\rho^{-2}\underbrace{(\overline{V}-\overline{V}_{MP})^{2}}_{O(1)}+(\log\varepsilon)^{-2}\rho^{-2}\underbrace{(W-W_{MP})^{2}}_{O(1)}\right),
\end{equation}The first four terms converge to zero using finite energy of $\Psi$ and $\Psi_{MP}$ and the DCT. The last two terms also converge to zero by the DCT and the fact that 
\begin{equation}
\int_{\mathcal{W}^2_{\delta,\varepsilon}}(\log\varepsilon)^{-2}\rho^{-2}=O\left(\left(\log\varepsilon\right)^{-1}\right)\to 0
\end{equation}Similarly, $I_3\to 0$. Furthermore, combining $\zeta^i-\zeta^i_{MP}=O(1)$, $-6h_1-6U-h_2-V, -6h_1-6U+h_2+V=O(1)$, and the finite energy of $\Psi$ and $\Psi_{MP}$, we obtain that $I_4,I_5\to 0$.
\end{proof}

We compose the three cut and paste operations defined above and define the following map
\begin{equation}\label{105}
	{\Psi}_{\delta,\varepsilon}
	=G_{\varepsilon}\left(F_{\delta}\left(
	\overline{F}_{\delta}(\overline{\Psi})\right)\right).
\end{equation}Then Lemma \ref{cutandpaste1}-Lemma \ref{cutandpaste3} lead to the following result.

\begin{prop}\label{proposition}
	Let $\varepsilon\ll\delta\ll 1$ and suppose that ${\Psi}$ satisfies the hypotheses of Theorem \ref{infimum}. Then
	${\Psi}_{\delta,\varepsilon}$ satisfies \eqref{54} and
	\begin{equation}\label{106}
		\lim_{\delta\rightarrow 0}\lim_{\varepsilon\rightarrow 0}
		\mathcal{I}({\Psi}_{\delta,\varepsilon})=\mathcal{I}({\Psi}).
	\end{equation}
\end{prop}
We are now in a position to establish the main result of this section.\medskip

\noindent\textit{Proof of Theorem \ref{infimum}.} Let $\tilde{\Psi}^{t}_{\delta,\varepsilon}$ be the geodesic connecting $\tilde{\Psi}_{MP}$ to
$\tilde{\Psi}_{\delta,\varepsilon}$ as described at the beginning of this section
\begin{equation}\label{107}
\frac{d^{2}}{dt^{2}}\mathcal{I}({\Psi}^{t}_{\delta,\varepsilon})
=
\underbrace{\frac{d^{2}}{dt^{2}}\mathcal{I}_{\mathcal{A}_{\delta,\varepsilon}}
	({\Psi}^{t}_{\delta,\varepsilon})}_{I_{1}}+\underbrace{\frac{d^{2}}{dt^{2}}\mathcal{I}_{\Omega_{\delta,\varepsilon}}
	({\Psi}^{t}_{\delta,\varepsilon})}_{I_{2}}.
\end{equation}It follows from Proposition \ref{proposition} that
$\overline{U}^{t}_{\delta,\varepsilon}=\overline{U}_{MP}+t(\overline{U}_{\delta,\varepsilon}-\overline{U}_{MP})$ and $\overline{V}^{t}_{\delta,\varepsilon}=\overline{V}_{MP}$ on $\mathcal{A}_{\delta,\varepsilon}$. Combining these with $\operatorname{dist}_{SL(3,\mathbb{R})/SO(3)}({\Psi}_{\delta,\varepsilon},{\Psi}_{0})
=12|\overline{U}_{\delta,\varepsilon}-\overline{U}_{0}|$ on $\mathcal{A}_{\delta,\varepsilon}$ and the asymptotes  \eqref{l1}-\eqref{fall4.3}, we can pass the derivative into the integral and similar to \cite{Alaeemassam} obtain that
\begin{align}\label{109}
\begin{split}
I_{1}\geq &2\int_{\mathcal{A}_{\delta,\varepsilon}}
|\nabla\operatorname{dist}_{SL(3,\mathbb{R})/SO(3)}
({\Psi}_{\delta,\varepsilon},{\Psi}_{MP})|^{2}
\end{split}
\end{align}
On domain $\Omega_{\delta,\varepsilon}$, using the relation of the reduced energy and harmonic energy and the convexity of the harmonic energy follows
\begin{align}\label{108}
\begin{split}
I_{1}=&\frac{d^{2}}{dt^{2}}E_{\Omega_{\delta,\varepsilon}}
(\tilde{\Psi}^{t}_{\delta,\varepsilon})-12\frac{d^{2}}{dt^{2}}\int_{\partial\Omega_{\delta,\varepsilon}
	\cap\partial\mathcal{A}_{\delta,\varepsilon}}
(h_1+2(\overline{U}_{MP}
+t(\overline{U}_{\delta,\varepsilon}-\overline{U}_{MP})-U_S)\partial_{\nu}\left(h_1+U_S\right)\\
&-\frac{d^{2}}{dt^{2}}\int_{\partial\Omega_{\delta,\varepsilon}
	\cap\partial\mathcal{A}_{\delta,\varepsilon}}(2\overline{V}_{MP}-h_2-V_S)\partial_{\nu}\left(h_{2}+V_S\right),
\\
\geq & 2\int_{\Omega_{\delta,\varepsilon}}
|\nabla\operatorname{dist}_{SL(3,\mathbb{R})/SO(3)}
({\Psi}_{\delta,\varepsilon},{\Psi}_{0})|^{2}.
\end{split}
\end{align}Thus we have
\begin{equation}\label{2ndvariation}
\frac{d^{2}}{dt^{2}}\mathcal{I}({\Psi^{t}}_{\delta,\varepsilon})
\geq 2\int_{\mathbb{R}^{3}}|\nabla\operatorname{dist}_{SL(3,\mathbb{R})/SO(3)}
({\Psi}_{\delta,\varepsilon},{\Psi}_{MP})|^{2}dx.
\end{equation}

We now verify that the first variation for ${\Psi}_{\delta,\varepsilon}$ vanishes. Choose $\varepsilon_{0}<\varepsilon$,
$\delta_{0}<\delta$ and write
\begin{equation}\label{110}
\frac{d}{dt}\mathcal{I}({\Psi}^{t}_{\delta,\varepsilon})
=\underbrace{\frac{d}{dt}\mathcal{I}_{\Omega_{\delta_{0},\varepsilon_{0}}}({\Psi}^{t}_{\delta,\varepsilon})}_{I_{3}}+
\underbrace{\frac{d}{dt}\mathcal{I}_{\mathcal{A}_{\delta_{0},\varepsilon_{0}}}({\Psi}^{t}_{\delta,\varepsilon})}_{I_{4}}.
\end{equation}
Using the relation of reduced energy and harmonic energy we have
\begin{equation}\label{111}
I_{3}=\frac{d}{dt}E_{\Omega_{\delta,\varepsilon}}
(\tilde{\Psi}^{t}_{\delta,\varepsilon})-\sum_{\pm}\int_{\partial B_{\delta_{0}}(p_{\pm})}24(\overline{U}_{\delta,\varepsilon}-\overline{U}_{MP})\partial_{\nu}\overline{U}_{MP}
-\int_{\partial\mathcal{C}_{\delta_{0},\varepsilon_{0}}}
24(\overline{U}_{\delta,\varepsilon}-\overline{U}_{MP})\partial_{\nu}\overline{U}_{MP}
\end{equation}where $\nu$ is the unit outer normal pointing towards the designated asymptotically flat end. Since $\overline{U}+|\nabla \overline{U}|$ is uniformly bounded and $\tilde{\Psi}_{MP}$ is a critical point of $E$, at $t=0$, the integral $I_3$ converges to zero as $\epsilon_0\to 0$ followed by $\delta_0\to 0$.  Next, using that
$\overline{U}^{t}_{\delta,\varepsilon}=\overline{U}_{0}+t(\overline{U}_{\delta,\varepsilon}-\overline{U}_{0})$
and $\frac{d}{dt}\overline{V}^{t}_{\delta,\varepsilon}=
\frac{d}{dt}W^{t}_{\delta,\varepsilon}
=\frac{d}{dt}\zeta^{1,t}_{\delta,\varepsilon}=\frac{d}{dt}\zeta^{2,t}_{\delta,\varepsilon}=0$ on $\mathcal{A}_{\delta_{0},\varepsilon_{0}}$ produces
\begin{align}\label{112}
\begin{split}
I_{4}=&O(t)+\int_{\mathcal{A}_{\delta_{0},\varepsilon_{0}}}
24\nabla \overline{U}_{MP}\cdot\nabla(\overline{U}_{\delta,\varepsilon}-\overline{U}_{MP})
-6(\overline{U}_{\delta,\varepsilon}-\overline{U}_{MP})
\frac{e^{-6h_1-6{U}_{\delta,\varepsilon}^{t}-h_2-{V}_{MP}}}{\cosh W_{MP}}|\nabla\zeta^{1}_{MP}|^{2}\\
&-\int_{\mathcal{A}_{\delta_{0},\varepsilon_{0}}}
6(\overline{U}_{\delta,\varepsilon}-\overline{U}_{MP})e^{-6h_1-6{U}_{\delta,\varepsilon}^{t}+h_2+{V}_{MP}}\cosh W_{MP}
|e^{-h_2-{V}_{MP}}\tanh W_{MP}\nabla\zeta^{1}_{MP}-\nabla\zeta^{2}_{MP}|^{2}.
\end{split}
\end{align}We can integrate by parts the first term and get the boundary terms in \eqref{111}. Then by reduced energy of $\Psi_{MP}$ and the DCT, similar to $I_3$, the integral $I_4$ converges to zero as $\epsilon_0\to 0$ followed by $\delta_0\to 0$.

Now integrating \eqref{2ndvariation} twice, applying a Sobolev inequality, a triangle inequality, and Proposition \ref{proposition} we obtain 
\begin{equation}\label{e1}
\begin{split}
\mathcal{I}(\Psi)-\mathcal{I}(\Psi_{MP})
\geq& C\left(\int_{\mathbb{R}^{3}}
\operatorname{dist}_{SL(3,\mathbb{R})/SO(3)}^{6}(\Psi,\Psi_{MP})dx
\right)^{\frac{1}{3}}\\
&-C\lim_{\delta\rightarrow 0}\lim_{\varepsilon\rightarrow 0}\int_{\mathbb{R}^{3}}
\operatorname{dist}_{SL(3,\mathbb{R})/SO(3)}^{6}({\Psi}_{\delta,\varepsilon},{\Psi})dx.
\end{split}
\end{equation}

In order to complete the proof, we need to show that last term vanishes. Combining the
triangle inequality and the fact that the distance between two points in $SL(3,\mathbb{R})/SO(3)$ is not greater than the length of a coordinate line connecting them, we have
\begin{align}\label{119}
\begin{split}
&\operatorname{dist}_{SL(3,\mathbb{R})/SO(3)}({\Psi}_{\delta,\varepsilon},{\Psi})\\
\leq&\operatorname{dist}_{SL(3,\mathbb{R})/SO(3)}
((\overline{U}_{\delta,\varepsilon},\overline{V}_{\delta,\varepsilon},W_{\delta,\varepsilon},
\zeta^{1}_{\delta,\varepsilon},\zeta^{2}_{\delta,\varepsilon}),
(\overline{U},\overline{V}_{\delta,\varepsilon},W_{\delta,\varepsilon},
\zeta^{1}_{\delta,\varepsilon},\zeta^{2}_{\delta,\varepsilon}))\\
&+\operatorname{dist}_{SL(3,\mathbb{R})/SO(3)}
((\overline{U},\overline{V}_{\delta,\varepsilon},W_{\delta,\varepsilon},
\zeta^{1}_{\delta,\varepsilon},\zeta^{2}_{\delta,\varepsilon}),
(\overline{U},\overline{V},W_{\delta,\varepsilon},
\zeta^{1}_{\delta,\varepsilon},\zeta^{2}_{\delta,\varepsilon}))\\
&+\operatorname{dist}_{SL(3,\mathbb{R})/SO(3)}
((\overline{U},\overline{V},W_{\delta,\varepsilon},
\zeta^{1}_{\delta,\varepsilon},\zeta^{2}_{\delta,\varepsilon}),
(\overline{U},\overline{V},W,
\zeta^{1}_{\delta,\varepsilon},\zeta^{2}_{\delta,\varepsilon}))\\
&+\operatorname{dist}_{SL(3,\mathbb{R})/SO(3)}
((\overline{U},\overline{V},W,
\zeta^{1}_{\delta,\varepsilon},\zeta^{2}_{\delta,\varepsilon}),
(\overline{U},\overline{V},W,
\zeta^{1},\zeta^{2}_{\delta,\varepsilon}))\\
&+\operatorname{dist}_{SL(3,\mathbb{R})/SO(3)}
((\overline{U},\overline{V},W,
\zeta^{1},\zeta^{2}_{\delta,\varepsilon}),
(\overline{U},\overline{V},W,
\zeta^{1},\zeta^{2}))\\
\leq& C\left(|\overline{U}-\overline{U}_{\delta,\varepsilon}|+|\overline{V}-\overline{V}_{\delta,\varepsilon}|
+|W-W_{\delta,\varepsilon}|\right)\\
&+Ce^{-3{U}-3h_1}\left(e^{-\tfrac{1}{2}{V}-\tfrac{1}{2}h_2}
|\zeta^{1}-\zeta^{1}_{\delta,\varepsilon}|
+e^{\tfrac{1}{2}{V}+\tfrac{1}{2}h_2}
|\zeta^{2}-\zeta^{2}_{\delta,\varepsilon}|\right).
\end{split}
\end{align}Similar to \cite[Theorem 4.1.]{Alaeemassam}, combing this with the asymptotes in Section \ref{sec:asymp}, the limit on the right hand side of \eqref{e1} vanishes. This complete the proof.
\section{Proof of the Theorem \ref{mainthm}}
Let $\Psi$ be the associate map of initial data set $(M,g,k)$ in the Theorem \ref{mainthm}. The asymptotic assumptions on the initial data $(M,g,k)$ imply that $\Psi$ satisfy the asymptotics
\eqref{fall1}-\eqref{fall4.2}. Since $\Psi_{MP}$ satisfies asymptotes of Section \ref{sec:asymp} as shown in Appendix \ref{AppMP}. Moreover, its mass, angular momentum, and area satisfy the following relation
\begin{equation}
\begin{split}
m_{MP}&=\left(\frac{3\pi}{8}\left(\frac{3A_{MP}}{16\pi}\right)^2
+\frac{\frac{3\pi}{8}\left(\frac{9}{4}\right)^{2}\mathcal{J}_{1}^{2}\mathcal{J}_{2}^{2}}
{\left(\frac{3A_{MP}}{16\pi}\right)^2}
+\frac{27\pi}{32}(\mathcal{J}_{1}^{2}+\mathcal{J}_{2}^{2})\right)^{1/3}\\
&=\mathcal{I}(\Psi_{MP})+\frac{\pi}{4}\int_{H}\left(2\overline{\alpha}_{MP}+6\overline{U}_{MP}-\text{sgn}(z)\overline{V}_{MP}\right)dz+m_S
\end{split}
\end{equation}Combining this with Theorem \ref{infimum} and \eqref{mass}, we obtain the inequality \eqref{mainineq}. The proof of the rigidity case is similar to \cite[Proof of Theorem 1.1.]{Alaeemassam}. %{\bfseries give details?}
 \hfill\qedsymbol\medskip

\appendix
\section{Schwarzschild-Tangherlini Weyl data}
The metric in standard exterior coordinates of  the five-dimensional Schwarzschild-Tangherlini Ricci flat black hole solution is 
\begin{equation}
\overline{g}=-\left(1-\frac{4\mathfrak{m}}{r^2}\right) dt^2+\left(1-\frac{4\mathfrak{m}}{r^2}\right)^{-1}dr^2+r^2\left(d\theta^2+\sin^2\theta (d\phi^1)^2+\cos^2\theta (d\phi^2)^2\right)
\end{equation} The transformation to Weyl coordinates is given by \cite{Pomeransky:2005sj}
\begin{eqnarray}
 \rho=\frac{1}{2}\sqrt{r^4-4\mathfrak{m}r^2}\sin2\theta, \qquad z=\frac{1}{2}\left(r^2-2\mathfrak{m}\right)\cos2\theta.
\end{eqnarray}Then the metric takes the form \eqref{GBmetric}
\begin{equation}
\overline{g}=-\frac{\mu^{-\mathfrak{m}}_-}{\mu^{\mathfrak{m}}_-}dt^2+f_{S}\left(d\rho^2+dz^2\right)+\mu^{\mathfrak{m}}_- (d\phi^1)^2+\mu^{-\mathfrak{m}}_+ (d\phi^2)^2
\end{equation} where
\begin{equation}
f_S =  \frac{\mu_-^\mathfrak{m} (\rho^2 + \mu_-^{-\mathfrak{m}}\mu_-^\mathfrak{m})}{(\rho^2 + (\mu_-^{-\mathfrak{m}})^2)(\rho^2 + (\mu_-^{\mathfrak{m}})^2)}
\end{equation}
where for $s=\mathfrak{m},-\mathfrak{m}$ we have 
\begin{eqnarray}
\mu^s_{\pm}=\sqrt{\rho^2+(z-s)^2}\pm(z-s)\qquad \Delta\log\mu^s_{\pm}=0\quad \text{on $\R^3-\Gamma$} . 
\end{eqnarray}
The initial data is parameterized by the functions
\begin{equation}
U_S=\frac{1}{4}\log\left(\frac{\mu^\mathfrak{m}_-\mu^{-\mathfrak{m}}_+}{\rho^2}\right),\qquad V_{S}=\frac{1}{2}\log\left(\frac{\rho^2}{\mu^{-\mathfrak{m}}_+\mu^\mathfrak{m}_-}\right)=-2U_S,\qquad W_S=0\label{U_S}
\end{equation}
and 
\begin{equation}
\alpha_S=\frac{1}{2}\log(2\sqrt{\rho^2+(z-\mathfrak{m})^2}) - U_S + \frac{1}{2} \log f_S
\end{equation} These functions have the property that $\mu_+^{\mathfrak{m}} \mu_-^\mathfrak{m} = \mu_+^{\mathfrak{-m}} \mu_-^\mathfrak{-m}  = \rho^2$. % and satisfy
%\begin{equation}
%\mu_+^\mathfrak{m}|_{\rho =0} = \begin{cases} 2(z-\mathfrak{m}) & z > \mathfrak{m} \\ 0 & z < \mathfrak{m} \end{cases}, \qquad \mu_+^{-\mathfrak{m}}|_{\rho =0} = \begin{cases} 2(z+\mathfrak{m}) & z > -\mathfrak{m} \\ 0 & z < -\mathfrak{m} \end{cases}
%\end{equation} and 
%\begin{equation}
%\mu_-^\mathfrak{m}|_{\rho =0} = \begin{cases} 0 & z > \mathfrak{m} \\ 2(\mathfrak{m}-z) & z < \mathfrak{m} \end{cases}, \qquad \mu_-^{-\mathfrak{m}}|_{\rho =0} = \begin{cases} -2(\mathfrak{m}+z) & z < -\mathfrak{m} \\ 0 & z > -\mathfrak{m}\end{cases}
%\end{equation} 
Then we have the following expansions on the horizon rod $H=(-\mathfrak{m},\mathfrak{m})$
\begin{eqnarray}
U_S&=&-\frac{1}{2}\log\rho+\frac{1}{4}\log\left(4|\mathfrak{m}^2-z^2|\right)+O(\rho^2)\label{Schw1}\\
V_{S}&=&\log\rho-\frac{1}{2}\log\left(4|\mathfrak{m}^2-z^2|\right)+O(\rho^2)\label{Schw2}\\
\alpha_S&=& \frac{1}{2} \log \rho + \frac{1}{2} \log |z-\mathfrak{m}| + \frac{1}{2} 
\log \mathfrak{m} - \frac{3}{4} \log {|\mathfrak{m}^2 - z^2|}+O(\rho^2)
\end{eqnarray} 
Outside the horizon rod $\Gamma_+\cup\Gamma_-$ we have the following expansions
\begin{eqnarray}
U_S&=&\frac{\text{sgn}(z)}{4}\log\left(\frac{|z+\mathfrak{m}|}{|z-\mathfrak{m}|}\right)+O(\rho^2)\label{Schw11}\\
V_{S}&=&-\frac{\text{sgn}(z)}{2}\log\left(\frac{|z+\mathfrak{m}|}{|z-\mathfrak{m}|}\right)+O(\rho^2)\label{Schw22}\\
\alpha_S&=&\frac{1}{2} \log \frac{ |z-\mathfrak{m}|}{\sqrt{|z^2 - \mathfrak{m}^2|}} + O(\rho^2)
\end{eqnarray} 
Now recall $r_+$ is the distance from the rod point $p_+=(0,\mathfrak{m})$ and $r_-$ is the distance from the rod point $p_-=(0, -\mathfrak{m})$, namely
\begin{equation}
r_{\pm}=\sqrt{\rho^2+(z\mp\mathfrak{m})^2}
\end{equation} which are related to the Weyl variables by
\begin{equation}
\rho^2 =\frac{1}{16\mathfrak{m}^2} \left(4r_+^2r_-^2-(r_+^2 + r^2_- -4\mathfrak{m}^2)^2\right),\qquad z = \frac{r_-^2 - r_+^2}{4\mathfrak{m}}
\end{equation}Near points $p_{\pm}$, we can define angles $\alpha_{\pm}$ such that $\rho=r_{\pm}\sin\alpha_{\pm}$ and $z=r_{\pm}\cos\alpha_{\pm}\pm \mathfrak{m}$. Then, we obtain the behaviour
\begin{equation} \begin{aligned}
e^{-4U_S} & = e^{2V_S} =  O(r_+), \qquad \text{as } r_+ \to 0, z > \mathfrak{m} \\ e^{-4U_S} & = e^{2V_S} =  O(\rho^2 r_+^{-1}), \qquad \text{as } r_+ \to 0, z < \mathfrak{m} \\
e^{-4U_S} & = e^{2V_S} =  O(r_-), \qquad \text{as } r_- \to 0, z <-
\mathfrak{m}  \\ e^{-4U_S} & = e^{2V_S} =  O(\rho^2 r_-^{-1}), \qquad \text{as } r_- \to 0, z > -\mathfrak{m}
\end{aligned}
\end{equation} Now notice that
\begin{equation}
e^{2\alpha} = \frac{\sqrt{ (r_+ + r_- - 2\mathfrak{m})(r_+ + r_- + 2\mathfrak{m})}}{2 r_-}
\end{equation} Then
\begin{equation}\begin{aligned}
e^{4\alpha_S} &= O(r_+), \quad \text{as } r_+ \to 0, z > \mathfrak{m} \qquad
e^{4\alpha_S} = O(\rho^2 r_+^{-1}), \quad \text{as } r_+ \to 0, z < \mathfrak{m} \\
e^{4\alpha_S} &= O(r_-^{-1}), \quad \text{as } r_- \to 0, z < -\mathfrak{m} \qquad
e^{4\alpha_S} = O(\rho^2 r_-^{-3}), \quad \text{as } r_- \to 0, z > -\mathfrak{m}
\end{aligned}
\end{equation} To work out the asymptotics for large $r$, it is useful to have the simple expressions
\begin{equation}
U_S = -\frac{1}{4} \log \left( 1- \frac{4\mathfrak{m}}{r^2}\right)
\end{equation} so that 
\begin{equation}
U_S = O(r^{-2}), \quad V_S = O(r^{-2}),  \quad \alpha_S = O(r^{-2}) \quad \text{as } \; r \to \infty
\end{equation} and
\begin{equation}
|\nabla U_S| = O(r^{-4}), \quad |\nabla V_S| = O(r^{-4}), \quad |\nabla \alpha_S| = O(r^{-4})
\end{equation}Finally, we express the ADM mass of the Schwarzschild-Tangherlini spacetime, following the steps in Section \ref{Sec1}, as follows 
\begin{eqnarray}\label{mS}
m_S&=&\frac{1}{16}\int_{\R^3}12|\nabla U_S|^{2}+|\nabla V_{S}|^{2} dx+\frac{\pi}{4}\int_{-\mathfrak{m}}^\mathfrak{m}\left[2\alpha_S+6U_S+V_{S}\right]dz+\pi \mathfrak{m}.
\end{eqnarray} 

%We can use the prolate spherical coordinate
%\begin{equation}
%x=\frac{r^2-2\mathfrak{m}}{2\mathfrak{m}},\quad y=\cos 2\theta,\qquad z=\mathfrak{m}xy,\qquad \rho^2=\mathfrak{m}^2(x^2-1)(1-y^2)
%\end{equation}
%Then we have
%\begin{equation}
%x=\frac{r_++r_-}{2\mathfrak{m}},\qquad y=\frac{r_+-r_-}{2\mathfrak{m}},\qquad r_{\pm}=\sqrt{\rho^2+(z\pm \mathfrak{m})^2}
%\end{equation}
\section{Myers-Perry Weyl data} \label{AppMP}
The three-parameter family of Myers-Perry vacuum black hole spacetimes \cite{Myers1986} are solutions to the vacuum Einstein equations in all dimensions greater than four, and contain event horizons with spatial sections of spherical topology. They can be regarded as the natural generalization to higher dimensions of the 4-dimensional Kerr black holes. In coordinates analogous to those of Boyer-Lindquist used for the Kerr solution, the 5-dimensional Myers-Perry metric takes the form
\begin{align}
	\begin{split}
		\mathbf{g} &= -dt^2+\frac{\kappa}{\Sigma}\left(dt+a\sin^2\theta d\phi^1+b\cos^2\theta d\phi^2\right)^2+\frac{\overline{r}^2\Sigma}{\Delta}d\overline{r}^2\\
		&+ \Sigma d\theta^2+\left(\overline{r}^2+a^2\right)\sin^2\theta (d\phi^1)^2
		+\left(\overline{r}^2+b^2\right)\cos^2\theta (d\phi^2)^2,
	\end{split}
\end{align}
where
\begin{gather}
	\Sigma=\overline{r}^2+b^2\sin^2\theta+a^2\cos^2\theta,\qquad \Delta=\left(\overline{r}^2+a^2\right)\left(\overline{r}^2+b^2\right)-\kappa \overline{r}^2.
\end{gather}
This family of solutions is parameterized by $(\kappa,a,b)$ which give rise to the mass and angular momenta through the formulae
\begin{gather}
	m=\frac{3}{8}\pi\kappa,\qquad \mathcal{J}_{1}=\frac{2}{3} ma,\qquad \mathcal{J}_{2}=\frac{2}{3} mb;
\end{gather}
where $\kappa\geq (a+b)^2$. Weyl coordinates are defined by
\begin{eqnarray}
\rho=\frac{1}{2}\sqrt{\Delta}\sin 2\theta,\qquad z=\frac{1}{4}\left(2\overline{r}^2+a^2+b^2-\kappa\right)\cos 2\theta
\end{eqnarray}  The set $\rho =0$ corresponds to the points on which the metric restricted to the Killing fields degenerates and consists of a finite horizon rod and two semi-infinite rods corresponds to fixed points sets of the rotational Killing fields $\partial_{\phi^1}, \partial_{\phi^2}$.  In particular the horizon rod  $H=(-\mathfrak{m},\mathfrak{m})$ where
\begin{equation}
\mathfrak{m}=\frac{1}{4}\sqrt{(\kappa-a^2-b^2)^2-4a^2b^2}
\end{equation} The black hole is referred to as extreme if $\kappa=(a+b)^2$; in this case the surface gravity vanishes and the black hole is degenerate. 

To find the inverse map it is useful to introduce the distance functions in $\mathbb{R}^2$ from the horizon endpoints $p_\pm$,
\begin{equation}
r_{\pm}=\sqrt{\rho^2+(z\mp\mathfrak{m})^2}
\end{equation} which can be expressed as
\begin{equation}
r_-  = \frac{\mu_+^{-\mathfrak{m}} + \mu_-^{-\mathfrak{m}}}{2}, \qquad r_+ = \frac{\mu_+^{\mathfrak{m}} + \mu_-^{\mathfrak{m}}}{2}
\end{equation}  We then have the inverse transformation
\begin{equation}
\bar{r} = \left[ r_- + r_+ + \frac{\kappa -a^2 - b^2}{2}\right]^{1/2}, \qquad \cos 2\theta = \frac{r_- - r_+}{2\mathfrak{m}}
\end{equation} This produces the relations
\begin{equation}
\rho^2 = \frac{1}{16\mathfrak{m}^2} \left(4\mathfrak{m}^2 - (r_--r_+)^2\right)\left((r_+ + r_-)^2 - 4 \mathfrak{m}^2\right), \qquad z = \frac{r_-^2 - r_+^2}{4\mathfrak{m}} . 
\end{equation} We may then express the metric of the slice $t=$constant of the Myers-Perry geometry in the form
\begin{equation}
h=\frac{e^{2U+2\alpha}}{2\sqrt{\rho^2+(z-\mathfrak{m})^2}}(d\rho^2+dz^2)+e^{2U}\lambda_{ij}d\phi^id\phi^j
\end{equation}
where
\begin{eqnarray}
\Sigma&=& r_++r_-+\left(\frac{a^2-b^2}{4\mathfrak{m}}\right)(r_--r_+)+\frac{\kappa}{2}\\
e^{4U}&=&\frac{4\mathfrak{m}^2-(r_+-r_-)^2}{16\mathfrak{m}^2\rho^2\Sigma}\left[\Sigma\left(r_++r_-+\frac{\kappa+a^2-b^2}{2}\right)\left(r_++r_-+\frac{\kappa-a^2+b^2}{2}\right)\right.\nonumber\\
&+&\left.\kappa\left(r_++r_-+\frac{\kappa-a^2-b^2}{2}\right)\left(\frac{a^2+b^2}{2}+\frac{b^2-a^2}{4\mathfrak{m}}(r_--r_+)\right)+\kappa a^2b^2\right]\\
e^{2\alpha}&=&e^{-2U}\frac{\Sigma}{2r_-}\\
\lambda_{11}&=&e^{-2U}\frac{(r_+-r_-)+2\mathfrak{m}}{8\mathfrak{m}}\left[2(r_++r_-)+\kappa+a^2-b^2+\frac{a^2\kappa\left[(r_+-r_-)+2\mathfrak{m}\right]}{2\mathfrak{m}\Sigma}\right]\\
\lambda_{22}&=&e^{-2U}\frac{(r_--r_+)+2\mathfrak{m}}{8\mathfrak{m}}\left[2(r_++r_-)+\kappa-a^2+b^2+\frac{b^2\kappa\left[(r_--r_+)+2\mathfrak{m}\right]}{2\mathfrak{m}\Sigma}\right]\\
\lambda_{12}&=&e^{-2U}\frac{ab\kappa\left[4\mathfrak{m}^2-(r_+-r_-)^2\right]}{16\mathfrak{m}^2\Sigma}
\end{eqnarray}
and
\begin{equation}
V=\frac{1}{2}\log\left(\frac{\lambda_{11}\mu^\mathfrak{m}_+}{\lambda_{22}\mu^\mathfrak{m}_-}\right),\qquad W=\sinh^{-1}\left(\frac{\lambda_{12}}{\rho}\right)
\end{equation}
where $\mu^\mathfrak{m}_{\pm}=\sqrt{\rho^2+(z-\mathfrak{m})^2}\pm (z-\mathfrak{m})$. Then we have the harmonic non-extreme Myers-Perry data $\Psi=(U,V,W,\zeta_i)$ where
\begin{eqnarray}
\zeta^1&=& \frac{a \kappa\left[\left(r_--r_+-2\mathfrak{m}\right)^2-6\mathfrak{m}^2\right]}{16\mathfrak{m}^2} - \frac{a\kappa(a^2-b^2)\left(r_--r_++2\mathfrak{m}\right)\left(r_+-r_-+2\mathfrak{m}\right)^2}{64\mathfrak{m}^3\Sigma} \\
\zeta^2&=&  -\frac{b \kappa\left[\left(r_--r_++2\mathfrak{m}\right)^2-6\mathfrak{m}^2\right]}{16\mathfrak{m}^2} - \frac{b \kappa(a^2-b^2)\left(r_--r_++2\mathfrak{m}\right)^2\left(r_+-r_-+2\mathfrak{m}\right)}{64\mathfrak{m}^3\Sigma}
\end{eqnarray}
Now we expand the non-extreme MP data $\Psi$ in each regions near the symmetry axis and rod points $p_\pm$. Let $\rho\to 0$ and $z\in H$, then 
\begin{equation}
\begin{split}
U&=-\frac{1}{2}\log\rho+\frac{1}{4}\log(|\mathfrak{m}^2-z^2|)+O(1),\\
V&= \log \rho - \frac{1}{2} \log\left( 4|z^2-\mathfrak{m}^2|\right) + O(1) \\
\alpha&=\frac{1}{2}\log\rho+\frac{1}{2}\log|z-\mathfrak{m}|-\frac{3}{4}\log(|\mathfrak{m}^2-z^2|) + \frac{1}{2} \log \mathfrak{m} +O(1),\\
W&=O(1),\\
\zeta_i&=O(1)
\end{split}
\end{equation}
Let $\rho\to 0$ and $z\in \Gamma_-\cup\Gamma_+$ we have
\begin{equation}
\begin{split}
U&=\begin{cases} \frac{1}{4} \log\left[ \frac{(\kappa + 4z)^2 - (a^2-b^2)^2 + 4 \kappa b^2}{16 (z^2-\mathfrak{m}^2)}\right] + O(1) & z > \mathfrak{m} \\
\frac{1}{4} \log\left[ \frac{(\kappa - 4z)^2 - (a^2-b^2)^2 + 4 \kappa a^2}{16 (z^2-\mathfrak{m}^2)}\right] + O(1)  & z< -\mathfrak{m} \end{cases} \\
V&=-\frac{\text{sgn}(z)}{2} \log\frac{|z+\mathfrak{m}|}{|z - \mathfrak{m}|} + O(\rho^2) \\
\alpha&=\begin{cases} -U + \frac{1}{2}\log\left(\frac{4z + \kappa + a^2-b^2}{4(z + \mathfrak{m})} \right) & z > \mathfrak{m} \\ -U + \frac{1}{2} \log \left(\frac{\kappa - 4z + b^2 -a^2}{4|z+\mathfrak{m}|}\right) & z < -\mathfrak{m} \end{cases} \\
W&=O(\rho)\\
\zeta_i&=O(1)
\end{split}
\end{equation} To investigate the asymptotics near the rod points, set $\rho = r_\pm \sin\alpha_\pm, z = r_\pm \cos\alpha_\pm \pm \mathfrak{m}$. Then
\begin{equation}
\begin{split}
e^{-4U} &=e^{2V} =  O(r_+), \qquad \text{as } r_+ \to 0, z > \mathfrak{m} \\
e^{-4U} & = e^{2V} =  O(\rho^2 r_+^{-1}), \qquad \text{as } r_+ \to 0, z < \mathfrak{m} \\
e^{-4U} & = e^{2V} = O(r_-), \qquad \text{as } r_- \to 0, z < -\mathfrak{m} \\
e^{-4U} &= e^{2V} = O(\rho^2 r_-^{-1}), \qquad \text{as } r_- \to 0, z > -\mathfrak{m}
\end{split}
\end{equation}  From this it can be similarly computed that as $r_+, r_- \to 0$ (in any direction)
\begin{equation}\begin{aligned}
\bar{U} &= O(1), \qquad \bar{V}  = O(1) \\
|\nabla \bar{U}  &= O(1), \qquad |\nabla \bar{V}| = O(1), \qquad |\nabla W| = O(1)
\end{aligned}
\end{equation} The asymptotic behaviour of $\alpha$ near the rod points is given by 
\begin{equation}\begin{aligned}
e^{4\alpha} &= O(r_+), \quad \text{as } r_+ \to 0, z > \mathfrak{m} \qquad
e^{4\alpha} = O(\rho^2 r_+^{-1}), \quad \text{as } r_+ \to 0, z < \mathfrak{m} \\
e^{4\alpha} &= O(r_-^{-1}), \quad \text{as } r_- \to 0, z < -\mathfrak{m} \qquad
e^{4\alpha} = O(\rho^2 r_-^{-3}), \quad \text{as } r_- \to 0, z > -\mathfrak{m}
\end{aligned}
\end{equation} We also have
\begin{equation} \begin{split}
\sinh^2 W & = O(\rho^2 r_+^{-1}), \quad  \text{as } r_+ \to 0, z > \mathfrak{m} \qquad \sinh^2 W  = O(r_+), \quad \text{as } r_+ \to 0, z < \mathfrak{m} \\
\sinh^2 W & = O(\rho^2 r_-^{-1}), \quad  \text{as } r_- \to 0, z > -\mathfrak{m} \qquad \sinh^2 W  = O(r_-), \quad \text{as } r_- \to 0, z < -\mathfrak{m}
\end{split}
\end{equation} Finally 
\begin{equation} \begin{split}
\zeta_1 & = -\frac{3}{8} a \kappa + \begin{cases}  O(\rho^4 r_+^{-2}) &  \text{as } r_+ \to 0, z > \mathfrak{m} \\  O(r_+^2) & \text{as } r_+ \to 0 , z < \mathfrak{m} \end{cases} \\
\zeta_1 & = \frac{5}{8} a \kappa + \begin{cases}  O(r_-) &  \text{as } r_- \to 0, z > -\mathfrak{m} \\  O(\rho^2 r_-^{-1}) & \text{as } r_+ \to 0 , z < \mathfrak{m} \end{cases} \\
\zeta_2 & = -\frac{5}{8} b \kappa + \begin{cases}  O(\rho^2 r_+^{-1}) &  \text{as } r_+ \to 0, z > \mathfrak{m} \\  O(r_+) & \text{as } r_+ \to 0 , z < \mathfrak{m} \end{cases} \\
\zeta_2 & = \frac{3}{8} b \kappa + \begin{cases}  O(r_-^2) &  \text{as } r_- \to 0, z > -\mathfrak{m} \\  O(\rho^4 r_-^{-2}) & \text{as } r_- \to 0 , z < -\mathfrak{m} \end{cases} 
\end{split}
\end{equation} Finally we turn to the asymptotics as $\bar{r} \to \infty$.  It is easily checked that
\begin{equation}\begin{aligned}
e^{4U} &= 1 + \frac{\kappa}{\bar{r}^2} + O(\bar{r}^{-4}) \Rightarrow U = O(\bar{r}^{-2}), \\ e^{2V} &= 1 + \frac{a^2 - b^2 - 4\mathfrak{m}}{\bar{r}^2} + O(\bar{r}^{-4}) \Rightarrow V = O(\bar{r}^{-2}) \\
e^{4\alpha} &= 1 + \frac{(a^2-b^2 - 4\mathfrak{m})\cos 2\theta}{\bar{r}^2} + O(\bar{r}^{-4}) \Rightarrow \alpha = O(\bar{r}^{-2}) \\
\sinh W & = \frac{a b \kappa \sin 2\theta}{2 \bar{r}^4} \Rightarrow W = O(\bar{r}^{-4}).
\end{aligned}
\end{equation} 

 \section*{Declarations} 
\subsection*{Funding}  A. Alaee acknowledge the support of an AMS-Simons travel grant. H Kunduri acknowledges support from NSERC Discovery Grant RGPIN-2018-04887. 
\subsection*{Competing Interests} The authors have no financial or proprietary interests in any material discussed in this article.
\subsection*{Authors' contribution statements}  All authors contributed equally  to the research presented above.

\end{document}